\documentclass[lettersize,journal]{IEEEtran}
\usepackage{amsthm,amsmath,amssymb,mathtools}
\usepackage{graphicx}
\usepackage{tabularx}
\usepackage{subfigure}
\usepackage{xcolor}
\usepackage{cite}
\usepackage{url}
\usepackage{comment}
\usepackage[ruled,vlined,linesnumbered]{algorithm2e}

\newtheorem{proposition}{Proposition}
\newtheorem{corollary}{Corollary}[proposition]
\newtheorem{definition}{Definition}
\theoremstyle{remark}
\newtheorem*{remark}{Remark}
\newcommand{\tp}{\mathsf{T}}

\renewcommand{\tilde}{\widetilde}

\SetKw{KwIn}{Input:}
\SetKw{KwOut}{Output:}
\SetKw{KwAnd}{and}
\SetKw{KwBreak}{break}

\begin{document}

\title{Integrated Cyber-Physical Resiliency for Power Grids under IoT-Enabled Dynamic Botnet Attacks}

\author{
\IEEEauthorblockN{Yuhan Zhao, Juntao Chen, and Quanyan Zhu}
\thanks{This work is partially supported by Grant ECCS-1847056 and Grant ECCS-2138956 from National Science Foundation (NSF).}
\thanks{Y. Zhao and Q. Zhu are with the Department of Electrical and Computer Engineering, Tandon School of Engineering, New York University, Brooklyn, NY 11201. Email: \{yhzhao, qz494\}@nyu.edu }
\thanks{J. Chen is with the Department of Computer and Information Sciences, Fordham University, New York, NY 10023. E-mail: jchen504@fordham.edu}
}


\maketitle

\begin{abstract}
The wide adoption of Internet of Things (IoT)-enabled energy devices improves the quality of life, but simultaneously, it enlarges the attack surface of the power grid system. The adversary can gain illegitimate control of a large number of these devices and use them as a means to compromise the physical grid operation, a mechanism known as the IoT botnet attack. This paper aims to improve the resiliency of cyber-physical power grids to such attacks. Specifically, we use an epidemic model to understand the dynamic botnet formation, which facilitates the assessment of the cyber layer vulnerability of the grid. The attacker aims to exploit this vulnerability to enable a successful physical compromise, while the system operator's goal is to ensure a normal operation of the grid by mitigating cyber risks. We develop a cross-layer game-theoretic framework for strategic decision-making to enhance cyber-physical grid resiliency. The cyber-layer game guides the system operator on how to defend against the botnet attacker as the first layer of defense, while the dynamic game strategy at the physical layer further counteracts the adversarial behavior in real time for improved physical resilience. A number of case studies on the IEEE-39 bus system are used to corroborate the devised approach.
\end{abstract}

\begin{IEEEkeywords}
Cyber-physical grid resilience, Botnet attacks, Dynamic games, Cross-layer defense
\end{IEEEkeywords}

\section{Introduction}

With the ubiquitous adoption of advanced information and communication technologies (ICTs), electric power systems have evolved as complex cyber-physical energy systems in which the functions of cyber and physical power components are tightly coupled during their operation. Among the vast ICTs adopted in smart power systems, one primary class is the Internet of Things (IoT) which plays an essential role in filling the gap between controlling and monitoring electricity services and physical processes. Another feature of the modern grid, in particular in the distribution system, is that massive IoT-controlled high-power energy devices are penetrated, such as air conditioners, water heaters, and electric ovens, and these devices can be controlled remotely by Internet connections.

The widespread adoption of IoT devices improves the quality of life. However, it exposes the grid to vast cyber threats, raising significant cybersecurity concerns \cite{zografopoulos2021cyber,chen2022cross}. The insecurity of these devices is partially due to the fact that cybersecurity is not a top concern when they are designed and manufactured. 
The presence of existing or zero-day vulnerabilities \cite{meneghello2019iot,alladi2020consumer} make IoT devices more susceptible to cyberattacks, such as DoS/DDoS attack \cite{kumari2023comprehensive} and botnet attack \cite{antonakakis2017understanding}.
In addition, their limited onboard computation capabilities render them incapable of running sophisticated encryption and authentication mechanisms, which makes them easy to hack. In this work, we consider the attacker deceptively manipulating massive IoT-controlled energy devices to compromise normal grid operation. This attack is termed an IoT botnet attack or load altering attack \cite{soltan2018blackiot,amini2016dynamic,mohsenian2011distributed,huang2019not}.  
It has been demonstrated that such IoT botnet attacks can lead to severe outcomes to the power grid, such as load shedding and generator tripping \cite{huang2019not}.
To launch the IoT botnet attack, the attacker needs to compromise a collection of IoT-controlled energy devices and alter their operation status (on/off) in a coordinated fashion, as illustrated in Fig. \ref{Figure_CPES}. 

\begin{figure}[!t]
  \centering
    \includegraphics[width=0.75\columnwidth]{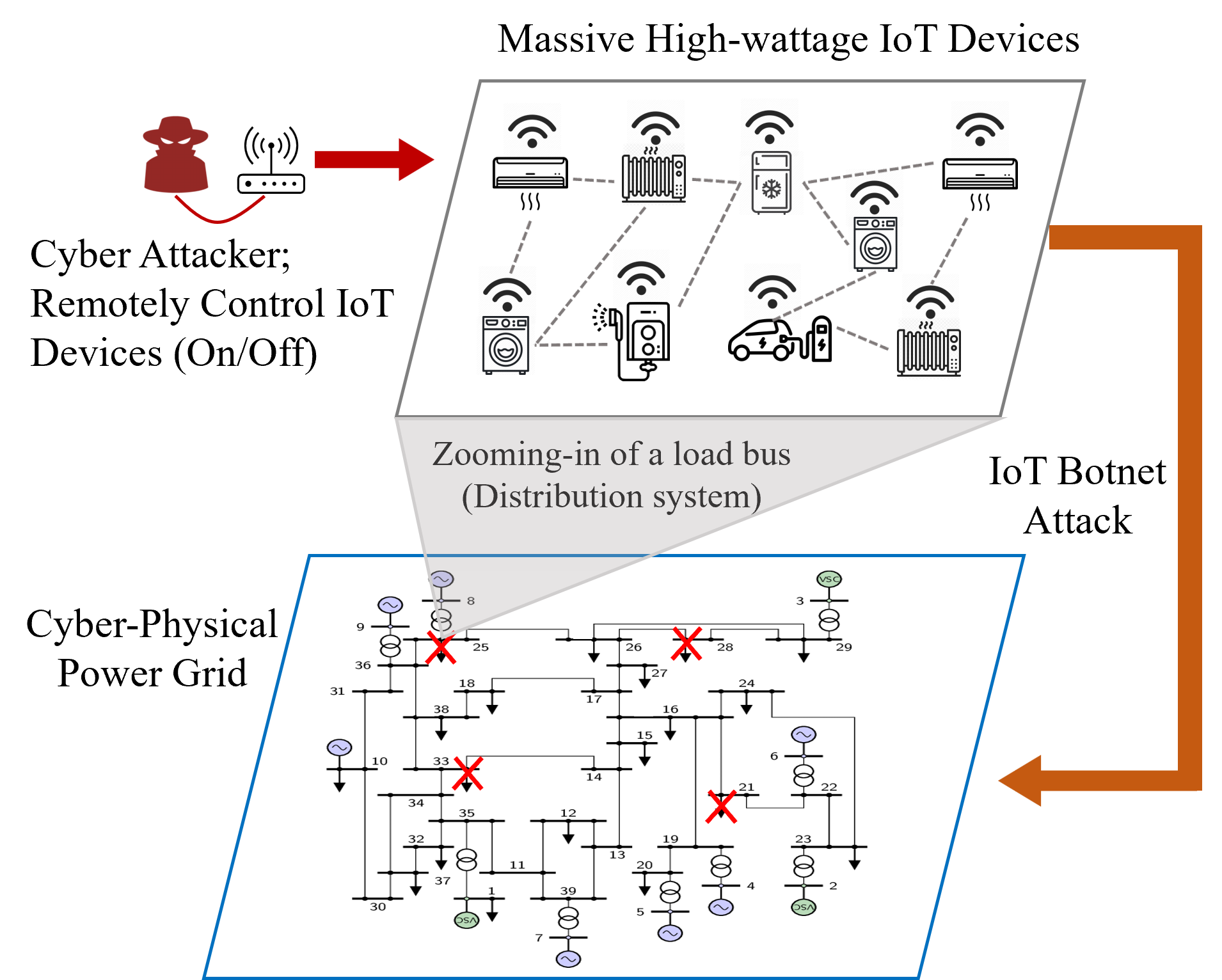}
  \caption[]{The massive IoT-controlled high-power energy devices introduce significant concerns on the cyber-physical power grid. Every load bus contains a considerable number of IoT devices. The adversary can manipulate these IoT devices in a coordinated fashion to launch a botnet attack that disrupts the grid operation.}
  \label{Figure_CPES}
\end{figure}

One essential consideration lacking in previous works is that they assumed that a certain proportion of the power load in the grid is vulnerable to being manipulated by the adversary and did not focus on how the attacker compromises these IoT devices and turns them into bots. It is imperative to understand the adversary's strategic behavior to develop effective countermeasures to protect the grid from IoT botnet attacks. The attacker generally uses malware to infect the IoT devices to form a botnet. The interconnection between IoT devices (connected directly or connected through the IoT platform) offers a convenient way for malware spreading. The attacker can increase the amount of load it can control in the grid to enable the attack. The grid system operator should take appropriate countermeasures to mitigate malware infections to enhance cyber resilience to IoT botnet attacks, such as patching the software, updating passwords, and reconfiguring the intrusion detection system of the IoT energy devices. 

In this work, we aim to develop a holistic approach to enhance the cyber-physical resiliency of the power grid under IoT botnet attack. Specifically, we first leverage a mean-field degree-based epidemic model to capture the dynamic malware propagation in the large-scale IoT device network. 
This modeling provides a macroscopic abstraction of compromising massive IoT devices in the distribution system and facilitates a quantitative assessment of the cyber vulnerability of the power grid under the botnet attack.
The outcome of the cyber layer directly impacts the feasibility of a successful IoT botnet attack. The attacker's goal is to destabilize the power grid (the transmission system) via strategic load manipulation at the physical layer while using the minimum amount of effort for malware propagation at the cyber layer. The grid operator (defender) needs to holistically devise an effective cyber defense strategy and resilient physical control to counteract the attacker's action in order to maintain safe grid operations. 

Capturing the multi-layer strategic interactions between the attacker and the defender requires a sophisticated game model. To this end, 
we develop a cross-layer game-theoretic framework in which the cyber layer game models the competition of malware spreading and detection in the botnet attack, and the physical dynamic game captures the interaction between the two players in controlling the grid dynamics. The solution concepts to both games are characterized by the Nash equilibrium (NE). The cyber layer game outcome directly impacts the players' strategies at the physical layer which calls for an integrated approach to enhance the grid's cyber-physical resiliency. The outcome strategy of the cyber game guides the system operator's defense against the botnet attacker which improves the cyber resilience of the power grid. In addition, the resilient control mechanism resulting from the dynamic game at the physical layer further counteracts the IoT-enabled botnet attack, which enhances the physical resilience of the grid.
Case studies on the IEEE-39 bus system show that the proposed cross-layer defensive scheme is effective in promoting the cyber-physical resiliency of the grid under IoT botnet attacks.

The contributions of this paper are summarized as follows.
\begin{enumerate}
    \item We develop a multi-layer game-theoretic approach to enhance the cyber-physical resiliency of power grids under IoT-enabled botnet attacks. 
    \item A tractable epidemic framework is proposed to quantitatively analyze the dynamic cyber risks of the grid under the botnet attack, which further facilitates the development of effective cyber-resiliency strategies.
    \item A dynamic game is established between the attacker and the grid operator at the physical layer, to which the NE strategy yields a strategic real-time counteraction to the dynamic botnet attack.
    \item We develop computationally efficient algorithms to find the NE solutions to both the cyber and physical layer games which are critical in developing holistic cyber-physical resilient schemes.
\end{enumerate}

The current work significantly differs from the preliminary version \cite{chen2023enhancing} in the following aspects. First, we provide a complete analysis of cyber risks of the cyber-physical grid (Section \ref{subsec:cyber_risk}), develop a cyber defense game framework, and further thoroughly characterize its equilibrium solution (Section \ref{subsec:cyber_game}). Second, the resilient control of the physical layer based on the dynamic game (Section \ref{sec:physical_design}) is completely new, while the preliminary version did not consider the strategic physical defense. Third, the case studies are greatly expanded based on the developed new framework and analytical results to showcase the cyber-physical cross-layer defense mechanism. Last but not least, the introduction and related works sections are significantly enriched.

\subsection{Organization of the Paper}
The rest of this paper is organized as follows. 
Section \ref{sec:related_work} discusses the related work.
Section \ref{sec:Prelim} introduces the basics of power grid dynamics and its counterpart under the IoT-enabled botnet attacks. Section \ref{sec:SIS} establishes an epidemic model to capture the cyber risk evolution of the power grids. Sections \ref{sec:cyber_design} and Section \ref{sec:physical_design} develop a holistic strategic cyber defense and robust physical operational schemes to enhance the power grid's cyber-physical resiliency. Case studies are used to demonstrate the proposed approaches in Section \ref{sec:case}, and Section \ref{sec:conclusion} concludes the paper.

\section{Related Work} \label{sec:related_work}
Cybersecurity is a practical concern to power systems as most of the adopted ICTs, such as phase measurement units (PMU), wide area measurement systems, and advanced metering infrastructure (AMI), are not built with strong security considerations. Cyberattacks on electric power systems can lead to undesired outcomes such as generator breakdown, power line outage, and even cascading failures, as demonstrated in the past \cite{sullivan2017cyber,bompard2013classification,langner2011stuxnet}. For example, in the infamous Ukraine power grid cyberattack on December 23, 2015, attackers leveraged a Trojan horse (malware) to intrude into the control system of the electric power grid and switched off breakers remotely, which led to the disconnections of 30 substations and approximately 225,000 consumers without electricity for 6 hours \cite{liang20162015}. Our work aims to improve the cybersecurity of modern power grids that exhibit a cyber-physical nature.

The integration of massive and heterogeneous IoT-operated high-power energy devices into the electric power grid gives rise to a new cybersecurity concern called IoT botnet attack \cite{soltan2018blackiot,shekari2021mamiot,dange2019iot}. 
The IoT botnet attack exploits a botnet (a group of compromised IoT devices) to execute rapid distributional attacks. In the infamous Mirai botnet attack, the attacker gained access to approximately 600,000 devices (e.g., routers, cameras) in a short period \cite{antonakakis2017understanding}, posing significant security threats to IoT networks. Recent studies corroborate the feasibility of IoT botnet attacks in the power systems, resulting in severe consequences like generator tripping and potential blackouts \cite{chen2020load,dabrowski2017grid,shekari2022madiot,lakshminarayana2022load}. Notably, these attacks can also extend to IoT-enabled high-power electric vehicle charging stations, which allows the attacker to manipulate substantial loads and disrupt both transmission and distribution systems in the power grid \cite{khan2019impact}.
%
While these previous works provide crucial insights into power system security to IoT botnet attacks, they are limited to analyzing the physical impacts of the attack on the grid operation.  Little attention was paid to the cyber layer analysis, such as attack modeling and dynamic risk assessment. This work aims to fill this gap by conducting a holistic understanding of the cyber-physical risks of the power grid under IoT botnet attacks and developing corresponding countermeasures.

To this end, we leverage epidemic models to study botnet attacks in the cyber layer. 
Initially developed for studying disease spread in populations, epidemic models have become prevalent for analyzing the propagation of phenomena within large-scale networks \cite{pastor2001epidemic,pastor2015epidemic}, providing an essential tool for cyber risk assessment. They have been applied to areas like malware dissemination in computer networks \cite{van2008virus,sahneh2013generalized}, the proliferation of fake news \cite{raponi2022fake,campan2017fighting}, and strategies for infection prevention and control \cite{nowzari2016analysis,omic2009protecting}. In power systems, epidemic models have been used to understand the dynamics of cascading failure and disturbances within the grid. For example, Wu et al. in \cite{wu2017disturbance} have used epidemic models to characterize the disturbance propagation in power systems resulting from factors like noisy data acquisition and transmission, and they have shown effectiveness in disturbance predictive accuracy compared with traditional analysis. Ma et al. in \cite{ma2012framework} have developed an epidemic-based framework to describe the frequency oscillation in power grids and have demonstrated the validity of their epidemic modeling approach using actual measurements from power system frequency monitoring networks. Furthermore, Zhang et al. in \cite{zhang2023analysis} have studied the propagation of failures across cyber and physical layers in power systems, and have developed a key node protection mechanism based on the Susceptible-Infected-Susceptible epidemic model, successfully tested on the IEEE 118-bus power system.

As cyberattacks frequently challenge power girds, much research has been focused on enhancing the grid's security and resiliency. The classical approach to protect the electric power system operation is through contingency plans ($N-1$ or $N-k$ contingency) \cite{luo2010rapid,wang2013two}. However, such a protection scheme may not be sufficient to counteract sophisticated attacks. 
More recent advances in protecting power system operations can be broadly categorized into two domains: cyber and physical. Within the cyber domain, countermeasures have been devised to defend various attacks, including false data injection attacks targeting the communication infrastructure of power systems \cite{musleh2019survey,liang2016review}, DoS attacks aimed at depleting power station resources \cite{liu2013denial,chen2019distributed}, and man-in-the-middle attacks seeking to maliciously manipulate grid data \cite{yang2012man,wlazlo2021man}.
On the physical side, attacks target assets like generators and loads. Defensive strategies have been studied to mitigate such threats, including methods to counteract load redistribution attacks \cite{yuan2011modeling,liu2014local,xiang2017game} and approaches to thwart load altering attacks \cite{mohsenian2011distributed,amini2016dynamic,lakshminarayana2021analysis}.
Different from the works that only focus on improving either the cyber or physical layer performance under attacks, our work develops a cross-layer approach to equip the power grid with integrated cyber-physical resiliency under strategic adversarial manipulations.

Game theory provides an effective framework for designing the cross-layer protection mechanism. It has also been widely used in both cybersecurity and power systems \cite{chi2021game,tushar2023survey,do2017game}. A major branch of game-theoretic studies in power systems predominantly addresses resource allocation for defense and protection.
For example, Wei et al. in \cite{wei2016stochastic} have formulated a coordinated cyber-physical attack protection strategy to stabilize the power grid based on stochastic games. The resulting Nash equilibrium allocates the defense resources to counteract the attacker without considering power flow or generator dynamics. Gao and Shi in \cite{gao2020modeling} have proposed a dynamic game-theoretic approach to mitigate cyber-physical attacks in power systems, where the defender and the attacker strategically allocate the resources to protect and compromise the cyber-physical elements in the power system using the Nash equilibrium. Similarly, Hasan et al. in \cite{hasan2020game} have developed a game-theoretic solution to strategically identify and protect critical stations in the power systems. Hyder and Govindarasu in \cite{hyder2020optimization} have studied the optimal investment strategies in the cybersecurity infrastructure of a smart grid based on game theory to deal with dynamically changing and uncertain adversary behaviors.
However, there is a gap in considering fine-grained attack and physical models alongside implementing a comprehensive cross-layer defensive mechanism against malicious attacks on large-scale IoT devices. We aim to bridge the gap in this work. 
Many works have also investigated the cyberattacks on the load side that aim to destabilize the power system. For example, dynamic load altering attacks and load redistribution attacks. To address dynamic load altering attacks, Amini et al. in \cite{amini2016dynamic} have proposed a protection scheme and PI controller design based on pole-placement optimization problems to stabilize the power grid. Eder et al. \cite{baron2017resilient} have developed a transactive control algorithm based on population games theory to redistribute the power demand and stable microgrids. Guo et al. in \cite{guo2021reinforcement} have leveraged minimax Q-learning to learn which generator bus to protect and prevent cascading failure caused by overload in transmission lines.
For load redistribution attacks, Xiang and Wang in \cite{xiang2017game} have leveraged game theory to seek the optimal defensive budget allocation for securing the power dispatch. Liu and Wang in \cite{liu2020defense} have focused on load redistribution attacks induced by insider threats. They have formulated the attack problem as a security resource allocation game and found the optimal strategy (e.g., load shedding) to protect the grid.
However, the methodologies in these works are insufficient to capture the cross-layer interdependency in the IoT botnet attack. Our work aims to develop a new approach to provide a holistic characterization of botnet attacks and improve the cyber-physical resiliency of the power grid.

\section{IoT-Enabled Botnet Attack to Power Grids} \label{sec:Prelim}

In this section, we first describe the system dynamics of the power grid under normal operation and then introduce the grid operation dynamics under IoT-enabled botnet attacks.

\subsection{Power Grid Base Model}
We consider a power grid consisting of a set of $\mathcal{N} = \mathcal{G} \cup \mathcal{L}$ buses, where $\mathcal{G} = \{g_1,\dots,g_{N_G} \}$ are generator buses and $\mathcal{L} = \{l_1,\dots,l_{N_L} \}$ are load buses, and $|\mathcal{N}|=N_G+N_L$. The power flow equations at the generator and load buses are given by \cite{amini2016dynamic}:
\begin{equation}
\label{eqn:powerflow_generator}
    P^G_i = \sum_{j \in \mathcal{G}} B_{ij} (\delta_i - \delta_j) + \sum_{j \in \mathcal{L}} B_{ij} (\delta_i - \theta_j), \quad \forall i \in \mathcal{G}, 
\end{equation}
\begin{equation}
\label{eqn:powerflow_load}
    -P^L_i = \sum_{j \in \mathcal{G}} B_{ij}(\theta_i - \delta_j) + \sum_{j \in \mathcal{L}} B_{ij} (\theta_i - \theta_j), \quad \forall i \in \mathcal{L}.
\end{equation}
Here, $P^G_i$ and $P^L_i$ are the power injection and consumption at bus $i$. $\delta_i$ and $\theta_j$ denote the voltage phase angle at the generator bus $i$ and the load bus $j$, respectively. $H_{ij}$ represents the admittance of the transmission line between buses $i$ and $j$. 

The generator dynamics at the generator bus $i$, $i \in \mathcal{G}$, can be modeled by the linear swing equation:
\begin{equation}
\label{eqn:generator_dyn}
\begin{split}
     \dot{\delta}_i &= \omega_i,\\
    M_i \dot{\omega}_i &= P^M_i - D^G_i \omega_i - P^G_i,
\end{split}
\end{equation}
where $\omega_i$ is the rotor frequency deviation at the generator bus $i$; $M_i$, $D^G_i$ and $P^M_i$ denote the rotor inertia, damping coefficient, and the mechanical power input, respectively.
Following \cite{amini2016dynamic,lakshminarayana2021analysis}, we assume that the mechanical power input is regulated by a proportional-integral (PI) controller, which is given by
\begin{equation}
\label{eqn:PI_controller}
    P^M_i = - K^P_i \omega_i - K^I_i \int_0^t \omega_i dt = - K^P_i \omega_i - K^I_i \delta_i,
\end{equation}
where $K^P_i > 0$ and $K^I_i > 0$ are the controller coefficients, respectively. When the load is identified, the power grid operator can design $K^P_i$ and $K^I_i$ for all generator buses to stabilize the power frequency deviation for safe operation. 
By integrating the controller \eqref{eqn:PI_controller} and the power flow equation \eqref{eqn:powerflow_generator} to the generate dynamics \eqref{eqn:generator_dyn}, we obtain for every generator bus $i \in \mathcal{G}$:
\begin{equation*}
\begin{split}
    -M_i \dot{\omega_i} = (K^P_i &+ D^G_i) \omega_i + K^I_i \delta_i \\ 
    &+ \sum_{j \in \mathcal{G}} B_{ij} (\delta_i - \delta_j) + \sum_{j \in \mathcal{L}} B_{ij} (\delta_i - \theta_j).
\end{split}
\end{equation*}
On the load side, we consider two types of loads in the grid \cite{zhao2014design}: the frequency-sensitive and frequency-insensitive loads. The former one at load bus $i$ can be represented by $D^L_i \phi_i$ where $\phi_i = -\dot{\theta}_i$ and $D^L_i > 0$ is the load damping coefficient. With a little abuse of notation, we use $P^L_i$ to represent the frequency-insensitive loads. Then, the power flow equation \eqref{eqn:powerflow_load} at the load bus $i$ can be rewritten as
\begin{equation*}
    -D^L_i \phi_i - P^L_i = \sum_{j \in \mathcal{G}} B_{ij}(\theta_i - \delta_j) + \sum_{j \in \mathcal{L}} B_{ij} (\theta_i - \theta_j).
\end{equation*}
Therefore, the overall power grid dynamics can be written as
\begin{equation}
\label{eqn:dyn}
\begin{split}
        & \begin{bmatrix} 
{\bf I} & {\bf O} & {\bf O} & {\bf O} \\
{\bf O}  & {\bf I}  & {\bf O} & {\bf O} \\
{\bf O} & {\bf O} & {-\bf M} & {\bf O} \\
{\bf O} & {\bf O} & {\bf O} & {\bf O}
\end{bmatrix}
\begin{bmatrix} 
\dot{\boldmath{\delta}} \\
\dot{\theta} \\
\dot{\omega} \\
\dot{\phi}
\end{bmatrix} = \begin{bmatrix} 
{\bf 0} \\
{\bf 0} \\
{\bf 0} \\
{\bf I} 
\end{bmatrix} \mathbf{P}^{L} 
+ \\
& \begin{bmatrix} 
{\bf O} & {\bf O} & {\bf I} & {\bf O} \\
{\bf O} &  {\bf O}  & {\bf O} & {\bf -I} \\
{\bf K}^{I} + {\bf B}^{GG} & {\bf B}^{GL} & {\bf K}^P + {\bf D}^G  & {\bf O} \\
{\bf B}^{LG} & {\bf B}^{LL} & {\bf O} &  {\bf D}^{L}
\end{bmatrix}
\begin{bmatrix} 
{\delta} \\
{\theta}   \\
{\omega} \\
\phi
\end{bmatrix}.
\end{split}
\end{equation}
Here, $\mathbf{P}^L \in \mathbb{R}^{N_L}$ denotes the aggregated load vector; $\delta\in \mathbb{R}^{N_G}, \omega \in \mathbb{R}^{N_G}$ are the aggregated phase angle and rotor frequency deviation at the generator buses, respectively;
$\theta\in \mathbb{R}^{N_L}, \phi \in \mathbb{R}^{N_L}$ denote the aggregated phase angle and the frequency deviation at the load buses, respectively; ${\bf M}\in \mathbb{R}^{N_G \times N_G},{\bf D}^G \in \mathbb{R}^{N_G \times N_G}$ and ${\bf D}^L \in \mathbb{R}^{N_L \times N_L}$ are diagonal matrices with diagonal entries given by the generator inertia, generator damping coefficients, and load damping coefficients, respectively;
${\bf K}^P \in \mathbb{R}^{N_G \times N_G},{\bf K}^I\in \mathbb{R}^{N_G \times N_G}$ are diagonal matrices with entries given by controller coefficients of the generators, respectively.  
We denote by ${\bf B}_{bus} = \begin{bmatrix} {\bf B}^{GG}  & {\bf B}^{GL} \\ {\bf B}^{LG} & {\bf B}^{LL} \end{bmatrix}$ the admittance matrix, where ${\bf B}^{GG} \in \mathbb{R}^{N_G \times N_G}, {\bf B}^{LL} \in \mathbb{R}^{N_L \times N_L}, {\bf B}^{GL} \in \mathbb{R}^{N_G \times N_L}$. 
The nominal frequency of the grid is denoted by $\omega_{\text{n}}$. The safety limits of the frequency deviation of generator bus $i\in\mathcal{G}$ satisfies $| \omega_{\text{n}} - \omega_i| \leq \omega_{\text{max}}$, where $\omega_{\text{max}}$ is the maximum permissible frequency deviation.

\subsection{Power Grid Model under Botnet Attacks}
Under IoT-based botnet attacks, the attacker manipulates the system load by synchronously switching on or off a large number of high-power devices. We follow the modeling in previous works \cite{amini2016dynamic,lakshminarayana2021analysis} by assuming that the frequency-insensitive loads at the load buses consist of two components ${\bf P}^L = {\bf P}^{LS} + {\bf P}^{LV}$, where ${\bf P}^{LS}$ and ${\bf P}^{LV}$ denote the secure portion of the load and the vulnerable portion of the load, respectively.
We denote $\mathcal{V}\subseteq \mathcal{L}$ as the set of vulnerable load buses, where   $|\mathcal{V}|=N_V$. Then, using botnet attacks, the attacker can manipulate some amount of the vulnerable load ${\bf P}^a \leq {\bf P}^{LV}$ in the power grid maliciously to disrupt and destabilize the normal operation. 

When a malicious attack happens, the existing PI controllers may not be sufficient to stabilize the grid. Besides, the attacker can launch strategic attacks by altering $\mathbf{P}^a$ dynamically, making it more challenging for predefined controllers to achieve the desired outcome. The grid operator requires additional mechanical power inputs, denoted as ${\bf P}^d \in \mathbb{R}^{N_G}$, to regulate the grid and mitigate the attack consequence. Therefore, the nominal power grid model \eqref{eqn:dyn} under botnet attacks can be modified to
\begin{equation}
\label{eqn:dyn_1}
\begin{split}
        & \begin{bmatrix} 
{\bf I} & {\bf O} & {\bf O} & {\bf O} \\
{\bf O}  & {\bf I}  & {\bf O} & {\bf O} \\
{\bf O} & {\bf O} & {-\bf M} & {\bf O} \\
{\bf O} & {\bf O} & {\bf O} & {\bf O}
\end{bmatrix}
\begin{bmatrix} 
\dot{\boldmath{\delta}} \\
\dot{\theta} \\
\dot{\omega} \\
\dot{\phi}
\end{bmatrix} = \begin{bmatrix} 
{\bf 0} \\
{\bf 0} \\
{\bf 0} \\
{\bf I} 
\end{bmatrix} (\mathbf{P}^{LS}+\mathbf{P}^a) 
+ \begin{bmatrix}  {\bf 0} \\ {\bf 0} \\ -{\bf I} \\ {\bf 0} \end{bmatrix} \mathbf{P}^{d} \\
& +\begin{bmatrix} 
{\bf O} & {\bf O} & {\bf I} & {\bf O} \\
{\bf O} &  {\bf O}  & {\bf O} & {\bf -I} \\
{\bf K}^{I} + {\bf B}^{GG} & {\bf B}^{GL} & {\bf K}^P + {\bf D}^G  & {\bf O} \\
{\bf B}^{LG} & {\bf B}^{LL} & {\bf O} &  {\bf D}^{L}
\end{bmatrix}
\begin{bmatrix} 
{\delta} \\
{\theta}   \\
{\omega} \\
\phi
\end{bmatrix}.
\end{split}
\end{equation}
The last row of \eqref{eqn:dyn_1} gives
\begin{equation}
\label{eqn:phi}
    \phi = -({\bf D}^L)^{-1} \left(
    \begin{bmatrix} 
        {\bf B}^{LG} \\  {\bf B}^{LL} \\ \bf{O}
    \end{bmatrix}^\tp 
    \begin{bmatrix}
        \delta \\ \theta \\ \omega
    \end{bmatrix}
     + {\bf P}^{LS} + {\bf P}^{a} \right).
\end{equation}
We substitute $\phi$ in \eqref{eqn:dyn_1} with \eqref{eqn:phi} and let $\mathbf{x} = \begin{bmatrix} \delta & \theta & \omega \end{bmatrix}$. Then we obtain the following dynamical system:
\begin{equation}
\label{eqn:dyn_2}
 \dot{\mathbf{x}} = A \mathbf{x} + B_d \mathbf{P}^d + B_a \mathbf{P}^a + c,
\end{equation}
where the state matrix $A \in \mathbb{R}^{2N_G+N_L}$ is given by
\begin{equation*}
    \begin{bmatrix}
        {\bf I} & {\bf O} & {\bf O} \\ 
        {\bf O} & ({\bf D}^L)^{-1} & {\bf O} \\ 
        {\bf O} & {\bf O} & -{\bf M}^{-1}
    \end{bmatrix} \cdot
    \begin{bmatrix}
        {\bf O} & {\bf O} & {\bf I} \\ 
        {\bf B}^{LG} & {\bf B}^{LL} & {\bf O} \\ 
        {\bf K}^I + {\bf B}^{GG} & {\bf B}^{GL}  & {\bf K}^P + {\bf D}^G
    \end{bmatrix};
\end{equation*}
the operator's and the attacker's input matrices and the constant term are given by
\begin{equation*}
    B_d = \begin{bmatrix}
        {\bf O} \\ {\bf O} \\ {\bf M}^{-1}
    \end{bmatrix}, \
    B_a = \begin{bmatrix}
        {\bf O} \\ ({\bf D}^L)^{-1} \\ {\bf O}
    \end{bmatrix}, \
    c = \begin{bmatrix}
        {\bf O} \\ ({\bf D}^L)^{-1} {\bf P}^{LS} \\ {\bf O}
    \end{bmatrix}.
\end{equation*}
Eq. \eqref{eqn:dyn_2} specifies the evolution of power system dynamics under defense and attack actions.
The corresponding discrete system with sampling time $T_s$ can be written as
\begin{equation}
\label{eqn:dyn_3}
    \mathbf{x}_{t+1} = \tilde{A} \mathbf{x}_t + \tilde{B}_d \mathbf{P}^d_t + \tilde{B}_a \mathbf{P}^a_t + \tilde{c},
\end{equation}
where the subscript $t$ represents the time step $t$.

\section{Epidemic Modeling of Botnet Attacks}\label{sec:SIS}
In this section, we model and quantify the systemic risk of power grids with a massive integration of IoT-operated energy devices using a susceptible-infected-susceptible (SIS) model. It has been empirically shown that such an epidemic model can capture the dynamics of botnet propagation with a high accuracy \cite{dagon2006modeling}. 

Due to the large-scale feature of IoT-controlled energy devices in the network, finite modeling would be prohibitive. We resort to complex network models, which can capture the characteristics of massive interconnections in the IoT device network. Specifically, let $k$ be the degree of an IoT device (can be regarded as a node) in the grid, where $k\in \mathcal{K}:=\{0,1,2,...,K\}$, and $p(k)\in[0,1]$ be the probability distribution of node's degree in the network. 
Further, we leverage the SIS epidemic model to characterize the stealthy botnet attack process and estimate the fraction of compromised IoT devices. Let $I_{k}(t)\in[0,1]$ be the density of the IoT devices of degree $k$ compromised by the attacker at time $t$. Then, the dynamics of the botnet attack process in the IoT device network can be described by \cite{pastor2015epidemic}:
\begin{equation}\label{eqn:Ik}
    \frac{dI_{k}(t)}{dt}=-\gamma I_{k}(t)+\zeta k[1-I_{k}(t)]\Theta(t),
\end{equation}
where $\gamma$ and $\zeta$ are the recovery and spreading rates, respectively;
\begin{equation}
    \Theta(t)=\frac{\sum_{k\in\mathcal{K}}kp(k)I_{k}(t)}{
\langle k \rangle}
\end{equation}
represents the probability of a given link connected to an infected IoT device; $\langle k \rangle =\sum_k kp(k)$ is the average connectivity of IoT devices in the grid. The attacker's behavior in the model is reflected by the spreading rate $\zeta$. A large $\zeta$ indicates that the attacker spends more effort compromising IoT-controlled energy devices.
The natural recovery rate $\gamma$ captures the malware elimination ability of IoT devices, e.g., via automatic software and firmware updates. 
In this work, we use a scale-free network to model the IoT device network in the power grid. Each IoT device is treated as a node and obeys a power-law degree distribution $p(k) \sim k^{-3}, k \in \mathcal{K}$. 

This epidemic model facilitates quantifying the level of vulnerable load and the systemic risks of the grid due to the integration of massive IoT-controlled energy devices. Specifically, the \emph{systemic risk} of the grid can be quantified by 
\begin{equation}\label{eqn:R}
    R(t) = I(t)\cdot N_d \cdot W_d,
\end{equation}
where $N_d$ and $W_d$ denote the estimated total number of IoT-controlled energy devices and their average power usage (watts), respectively. The quantity 
\begin{equation}\label{eqn:I}
    I(t)=\sum_{k\in\mathcal{K}} p(k)I_{k}(t)
\end{equation}
represents the aggregated percentage of compromised IoT-controlled energy devices via the botnet attack and can be used to quantify the \emph{cyber risk}. It can be observed that the systemic risk $R(t)$ is not static but evolves dynamically governed by the malware spreading process \eqref{eqn:Ik}. To quantify the risk propagation in terms of the attacker's behavior $\zeta$ and the underlying cyber dynamics, it is necessary to analyze the differential equation \eqref{eqn:Ik}, which is pursued in the next section.

\section{Risk Analysis and Cyber-Resilient Design}\label{sec:cyber_design}
The developed epidemic model provides a systematic approach for risk quantification of the power system to IoT botnet attack. After recruiting a certain level of bots, the attacker needs to determine how to deploy the attack in terms of the location and the scale to disrupt the power grid operation. The defender, on the other hand, should devise effective means to counteract the attack. To this end, a holistic cyber-physical analysis is imperative. This section will first analyze the cyber layer risks and then develop a game-theoretic approach to enhance the cyber resilience of the power grids.

\subsection{Cyber Risk Analysis}\label{subsec:cyber_risk}
The malware spreading dynamics \eqref{eqn:Ik} describes the changes in the scale of IoT energy devices that can be maliciously manipulated. This process can be seen as botnet recruitment, in which the attacker aims to turn benign devices into bots. 
Since the attacker can manipulate the vulnerable loads after completing the botnet recruitment (the IoT infection), the systemic risk $\bar{R}$ and the cyber risk $\bar{I}$ at steady states are of more interest. Therefore, we investigate the steady state of the dynamics \eqref{eqn:Ik} at which $I_k(t), k\in \mathcal{K}$, reaches an equilibrium. This can be found by imposing stationarity: 
\begin{equation*}
    \frac{dI_{k}(t)}{dt}=-\gamma I_{k}(t)+\zeta k[1-I_{k}(t)]\Theta(t) = 0,
\end{equation*}
which gives 
\begin{equation}
\label{eqn:Theta}
    \bar{I}_k = \frac{\zeta k \bar{\Theta}(\zeta,\gamma)}{\gamma+\zeta k \bar{\Theta}(\zeta,\gamma)},
\end{equation}
where $\bar{\Theta}(\zeta,\gamma)$ indicates that $\Theta(t)$ only depends on $\zeta$ and $\gamma$ at the steady state. Based on \eqref{eqn:Theta}, we obtain, at the steady state,
\begin{equation}\label{eqn:Theta_S}
     \bar{\Theta}(\zeta,\gamma)=\frac{\sum_{k\in\mathcal{K}}kp(k) \bar{I}_{k}}{
\langle k \rangle}.
\end{equation}
Substituting \eqref{eqn:Theta} into \eqref{eqn:Theta_S}, we have a self-consistency equation on $\Theta(\zeta,\gamma)$ and hence we can obtain the cyber risk $\bar{I}$. We have the following proposition to characterize epidemic-free conditions.

\begin{proposition} \label{prop:1}
For any $\gamma > 0$ and $\zeta > 0$, the equilibrium cyber risk $\bar{I}$ is unique in $(0,1]$ if $\frac{\gamma}{\zeta} < \frac{\langle k^2 \rangle}{\langle k \rangle}$ and is $0$ if $\frac{\gamma}{\zeta} \geq \frac{\langle k^2 \rangle}{\langle k \rangle}$, where $\langle k^2 \rangle = \sum_k k^2 p(k)$.
\end{proposition}

\begin{proof}
We can tell from \eqref{eqn:Theta_S} that $\bar{\Theta}(\gamma,\zeta) \in [0,1]$. The self-consistency equation on $\bar{\Theta}(\zeta,\gamma)$ yields
\begin{equation*}
    \bar{\Theta} = \frac{1}{\langle k \rangle} \sum_k \left( k p(k) \frac{\zeta k \bar{\Theta}}{\gamma + \zeta k \bar{\Theta}}  \right).
\end{equation*}
Note that $\bar{\Theta} = 0$ is always a feasible solution. Now assume $\bar{\Theta} \neq 0$, we have 
\begin{equation*}
    \sum_k k p(k) \frac{\zeta k}{\gamma + \zeta k \bar{\Theta}} - \langle k \rangle = 0.
\end{equation*}
Let $f(\Theta) = \sum_k k p(k) \frac{\zeta k}{\gamma + \zeta k \Theta} - \langle k\rangle$, which is continuous in $\Theta$. Note that $f'(\Theta) = -\sum_k kp(k) \frac{(\zeta k)^2}{(\gamma + \zeta k \Theta)^2} < 0$ for all $\Theta \in (0, 1]$ and $f(1) < 0$. Therefore, $f(\Theta)$ only has one zero if $f(0) > 0$, which gives $\frac{\zeta}{\gamma} \sum_k k^2 p(k) - \langle k \rangle > 0$, i.e., $\frac{\gamma}{\zeta} < \frac{\langle k^2 \rangle}{\langle k \rangle}$.
\end{proof}

Prop. \ref{prop:1} indicates that we can effectively reduce the cyber risk if the defender can protect the network from the malware up to a threshold. The threshold is related to the IoT network topology. If devices in the IoT network have more connections (larger $k$), the protection threshold becomes larger and the defender requires more effort to defend against the same level of the botnet attack.

The closed-form solution to $\bar{I}$ and $\bar{\Theta}$ for general degree distributions can be challenging to obtain. To characterize the cyber risk analytically, we use $k$ as a continuous random variable to approximate $\bar{I}_k$ and $\bar{\Theta}$ in \eqref{eqn:Theta}-\eqref{eqn:Theta_S}, which is valid for large-scale networks as in our scenario. The corresponding continuous power-law degree distribution becomes 
\begin{equation}\label{eqn:degree_dist}
    p(k) = 2d_{\min}^2 k^{-3}, \quad k \geq d_{\min},
\end{equation}
where $d_{\min}$ is the minimum number of connections of each node in the IoT device network. For example, $d_{\min}=2$ indicates that each IoT device is at least connected to two other devices. The average connectivity is
\begin{equation}
    \langle k \rangle = \sum_k kp(k) \simeq \int_{d_{\min}}^{\infty} k p(k) dk =  2d_{\min},
\end{equation}
Also, based on the continuous approximation of the degree $k$, we obtain
\begin{equation}
\begin{split}
     \bar{\Theta}(\zeta,\gamma) &= \int_{d_{\min}}^\infty d_{\min} k^{-2} \frac{\zeta k \bar{\Theta}(\zeta,\gamma)}{\gamma+\zeta k \bar{\Theta}(\zeta,\gamma)} dk \\
     & = {d_{\min}\zeta \bar{\Theta}(\zeta,\gamma)}  \int_{d_{\min}}^\infty \frac{1}{k(\gamma+\zeta k \bar{\Theta}(\zeta,\gamma))} dk,
 \end{split}
\end{equation}
which further gives
\begin{equation}
\label{eqn:Theta_S2}
    \bar{\Theta}(\zeta,\gamma) = \frac{\gamma}{\zeta d_{\min}} \cdot \frac{e^{-\gamma/(d_{\min}\zeta)}}{(1-e^{-\gamma/(d_{\min}\zeta)})}.
\end{equation}
The cyber risk $\bar{I}$ at the steady state becomes
\begin{equation}
\label{eqn:I_equilibrium}
\begin{split}
     \bar{I} &= \int_{d_{\min}}^\infty p(k) \bar{I}_k dk\\
     & = 2d_{\min}^2 \zeta \bar{\Theta}(\zeta,\gamma) \int_{d_{\min}}^\infty  k^{-2} \frac{1}{\gamma+\zeta k \bar{\Theta}(\zeta,\gamma)} dk.
     \end{split}
\end{equation}
Substituting \eqref{eqn:Theta_S2} into \eqref{eqn:I_equilibrium} and evaluating the integral leads to the following result: At the steady state, the percentage of infected IoT-enabled energy devices satisfies
\begin{equation}
\label{eqn:I_equ_final}
    \bar{I} 
    \sim e^{-\gamma/(d_{\min}\zeta)}.
\end{equation}
Based on \eqref{eqn:I_equ_final}, one can see that as the attack effort $\zeta$ increases, the size of the botnet (i.e., compromised energy devices) enlarges, which yields a higher level of systemic risk $\bar{R}=\bar{I}\cdot N_d \cdot W_d$ of the power grid.

\subsection{Cyber Defense Game}\label{subsec:cyber_game}
The attacker aims to devise a cost-effective malware spreading strategy (corresponding to large spreading $\zeta$) to disrupt the grid operation. While the system operator needs to enhance the cyber resilience of the grid to prevent the attack from happening. One way to achieve this is to improve the cyber detection capability of malware spreading. This countermeasure reduces the possibility of malicious control of IoT devices by the attacker. Some specific mechanisms include requiring the users to have manual software patches, regular password changes, etc., which corresponds to a higher recovery rate $\gamma$. Note that $\gamma$ and $\zeta$ reflect the consequence of the system operator (defender) and the attacker's effort to protect and attack the IoT device network. We capture the defender and attacker's cyber effort by $u_d$ and $u_a$, which affect the epidemic by $\gamma(u_d)$ and $\zeta(u_a)$.

The competition between the defender and the attacker constitutes a noncooperative game in which the attacker aims to compromise as many IoT-controlled energy devices as possible by botnet attack while the defender's objective is to reduce systemic risk through cyber defense. We define the cyber defense game as follows:
\begin{equation}
\label{eqn:cyber_game}
\begin{split}
    \min_{u_d \in \mathcal{U}_d} \quad & L_d(u_d, u_a):= C_d(u_d) + \bar{I}(u_d, u_a), \\ 
    \max_{u_a \in \mathcal{U}_a} \quad & L_a(u_d, u_a):= -C_a(u_a) + \bar{I}(u_d, u_a).
\end{split}
\end{equation}
Here, $\mathcal{U}_d \subseteq \mathbb{R}_+$ and $\mathcal{U}_a \subseteq \mathbb{R}_+$ are the admissible control sets; $C_d: \mathcal{U}_d \to \mathbb{R}_+$ and $C_a: \mathcal{U}_a \to \mathbb{R}_+$ denote the defense and attack cost. The cyber risk $\bar{I}$ in \eqref{eqn:I_equ_final} becomes a function of $u_d$ and $u_a$ because of $\gamma(u_d)$ and $\zeta(u_a)$.

We make some general assumptions on \eqref{eqn:cyber_game} to facilitate the subsequent analysis. First, we assume that $0 \in \mathcal{U}_d$ and $0 \in \mathcal{U}_a$, meaning that zero defense/attack effort is admissible. Second, we assume that the cost functions $C_d$ and $C_a$ are monotonically increasing and convex with $C_d(0) = 0$ and $C_a(0) = 0$. Third, we assume that $\gamma$ and $\zeta$ are monotonically increasing with $\gamma(0) > 0$ and $\zeta(0) > 0$. We also assume that the two functions are concave, indicating decreasing marginal effects in cyber protection and attack. This is because there are no perfect defenses and attacks in practice, even if one puts in a large protection and attack effort. 
We set $\bar{I}(u_d, u_a) = e^{-\gamma(u_d) / d_{\min} \zeta(u_a)}$ to obtain analytical results.

\subsubsection{Characterization of Cyber Risk $\bar{I}$}
We first analyze the property of the cyber risk $\bar{I}$ to approach the equilibrium solution of the cyber defense game \eqref{eqn:cyber_game}. An immediate corollary of Prop.~\ref{prop:1} is the following.

\begin{corollary}
For IoT networks with finite degree support set $\mathcal{K}$, the risk $\bar{I} = 0$ if $(u_d, u_a) \in \Phi(u_d, u_a) := \{ (u_d \in \mathcal{U}_d, u_a \in \mathcal{U}_a): u_d \geq \gamma^{-1}(\frac{\langle k^2 \rangle}{\langle k \rangle} \zeta(u_a)) \}$. For linear cases where $\gamma(u_d) = k_d u_d + \gamma_0$ and $\zeta(u_a) = k_a u_a + \zeta_0$, $\Phi(u_d, u_a) = \{(u_d, u_a): u_d \geq \frac{\langle k^2 \rangle k_a}{\langle k \rangle k_d} u_a + \frac{1}{k_d} (
\frac{\langle k^2 \rangle \zeta_0}{\langle k \rangle} - \gamma_0 )\}$, which is a convex region in $\mathcal{U}_d \times \mathcal{U}_a$.
\end{corollary}

The linear choices of $\gamma$ and $\zeta$ provide a well-defined convex region to characterize epidemic-free conditions. However, this is not true for general cases. Besides, we must look into the case where $\bar{I}$ is positive.
We can check that $\frac{\partial \bar{I}}{\partial u_d} < 0$ for all $u_d \geq 0$ and $\frac{\partial \bar{I}}{\partial u_a} > 0$ for all $u_a \geq 0$, meaning that increasing defense or attack effort can lead to a monotonically decreasing or increasing in the infectious rate. Further computation shows $\frac{\partial^2 \bar{I}}{\partial u_d^2} > 0$, which indicates that the defense effort has a decreasing marginal effect in reducing the cyber risk. Therefore, the defender's objective is convex in $u_d$ for any $u_a$, and thus provides a unique optimal response to $u_a$. 
However, we have
\begin{equation}
\label{eqn:d2I_dua2}
\begin{split}
    \frac{\partial^2 \bar{I}}{\partial u_a^2} &= e^{-\frac{\gamma}{d_{\min}\zeta}} (\frac{\gamma \zeta'}{d_{\min}\zeta^2})^2 + e^{-\frac{\gamma}{d_{\min}\zeta}} \frac{\gamma}{d_{\min}\zeta^3} (\zeta'' \zeta - 2(\zeta')^2) \\
    &= 
    e^{-\frac{\gamma}{d_{\min}\zeta}} \frac{\gamma}{d_{\min}\zeta^2} \left[ (\frac{\gamma}{d_{\min}} - 2\zeta) \frac{(\zeta')^2}{(\zeta)^2}  +\zeta'' \right],
\end{split}
\end{equation}
which indicates a complex behavior of $\bar{I}$ in the attacker's attack effort $u_a$ given $u_d$. We have the following proposition to characterize $\bar{I}$.

\begin{proposition}
\label{prop:2}
The cyber risk $\bar{I}(u_d, u_a)$ is concave in $u_a$ for $u_a \geq \zeta^{-1} \left( \frac{\gamma(u_d)}{2d_{\min}} \right)$ given an $u_d \in \mathcal{U}_d$. In addition, $\bar{I}$ becomes concave in $u_a \in \mathcal{U}_a$ for any $u_d$ if $ \zeta(0) \geq \frac{\gamma(u_d)}{2d_{\min}}$, $\forall u_d \in \mathcal{U}_d$.
\end{proposition}

\begin{proof}
Let $\frac{\gamma(u_d)}{d_{\min}}-2\zeta(u_a) \leq 0$, we obtain $u_a \geq \zeta^{-1} \left( \frac{\gamma(u_d)}{2d_{\min}} \right)$ because $\zeta' > 0$. Since $\zeta'' \leq 0$ from assumptions, we have $\frac{\partial^2 \bar{I}}{\partial u_a^2} \leq 0$ for $u_a \geq \zeta^{-1} \left( \frac{\gamma(u_d)}{2d_{\min}} \right)$. 
To strengthen the concavity condition for all $u_a \in \mathcal{U}_a$, we need $\zeta^{-1} \left( \frac{\gamma(u_d)}{2d_{\min}} \right) \leq 0$ for all $u_d \in \mathcal{U}_d$, i.e., $\zeta(0) \geq \frac{\gamma(u_d)}{2d_{\min}}$.
\end{proof}

Prop.~\ref{prop:2} provides a strengthened concavity condition in $\bar{I}$, which can be useful for deriving the NE of the cyber defense game \eqref{eqn:cyber_game}.
However, the condition is impractical because it indicates that the initial malware spreading rate $\zeta(0)$ with $u_a=0$ needs to be greater than the recovery rate, even if the defender has taken a defense effort $u_d > 0$. 
Besides, due to general forms of $\gamma$ and $\zeta$, the cyber risk $\bar{I}$ can exhibit different properties, such as convex-concave in $u_a$, depending on specific choices of $\gamma$ and $\zeta$. We use the following proposition to identify one general case of $\bar{I}$, which is valid for a wide class of $\gamma$ and $\zeta$, including linear function, power functions $x^\alpha (0 < \alpha < 1)$, and log functions $\log(x+1)$. 

\begin{proposition} \label{prop:3}
Let $\zeta(0) \geq \frac{\gamma(u_d)}{2d_{\min}}$ only for some $u_d \in \mathcal{U}_d$. We further assume that $\zeta''$ is differentiable. If $\frac{\partial^2 \bar{I}}{\partial u_a^2}\big\vert_{u_d, u_a=0} > 0$ for any $u_d \in \mathcal{U}_d$, then $\frac{\partial^2 \bar{I}}{\partial u_a^2}$ has a zero in $[0, \zeta^{-1}(\frac{\gamma}{2d_{\min}}))$. Furthermore, if $\zeta''' + \frac{2\zeta'}{\zeta^3}[(\frac{\gamma(u_d)}{d_{\min}}-2\zeta)\zeta\zeta'' - (\frac{\gamma(u_d)}{d_{\min}}-\zeta)(\zeta')^2] < 0$, the zero is unique, and $\bar{I}(u_d, u_a)$ is convex-concave in $u_a \in \mathcal{U}_a$ given any $u_d \in \mathcal{U}_d$.
\end{proposition}

\begin{proof}
From Prop.~\ref{prop:2}, $\bar{I}(u_d, u_a)$ is readily concave in $u_a \in \mathcal{U}_a$ for the $u_d$ such that $\zeta(0) \geq \frac{\gamma(u_d)}{2d_{\min}}$. 
For the rest $u_d \in \mathcal{U}_d$, let $f(u_a) = f_1(u_a) + f_2(u_a)$ where $f_1(u_a) = (\frac{\gamma}{d_{\min}}-2\zeta(u_a)) \frac{(\zeta'(u_a))^2}{(\zeta(u_a))^2}$ and $f_2(u_a) = \zeta''(u_a)$. $f$ is continuous in $u_a$. Since $f(0) > 0$ and $f(\zeta^{-1}(\frac{\gamma}{2d_{\min}})) < 0$, using the intermediate value theorem, $f(u_a)$ has at least one zero in $[0, \zeta^{-1}(\frac{\gamma}{2d_{\min}}))$.
Furthermore, we can check if
\begin{equation*}
    f'(u_a) = \zeta''' + \frac{2\zeta'}{\zeta^3}[(\frac{\gamma(u_d)}{d_{\min}}-2\zeta)\zeta\zeta'' - (\frac{\gamma(u_d)}{d_{\min}}-\zeta)(\zeta')^2] < 0,
\end{equation*}
$f$ is monotonically decreasing on $[0, \zeta^{-1}(\frac{\gamma}{2d_{\min}}))$ and the zero is unique, which indicates that $\frac{\partial^2 \bar{I}}{\partial u_a^2}$ only has one zero since $e^{-\frac{\gamma}{d_{\min}\zeta}} \frac{\gamma}{d_{\min}\zeta^2} > 0$.  
We write the zero as $\tilde{u}_{a,0}(u_d)$ to denote its dependency on $u_d$. For $u_d \leq \gamma^{-1}(2d_{\min}\zeta(0))$, i.e., $\zeta(0) \geq \frac{\gamma(u_d)}{2d_{\min}}$, we set $\tilde{u}_{a,0}(u_d) =\varnothing$ and $\bar{I}$ is concave in $u_a \in \mathcal{U}_a$; for $u_d > \gamma^{-1}(2d_{\min}\zeta(0))$, i.e., $\zeta(0) < \frac{\gamma(u_d)}{2d_{\min}}$, we have $\tilde{u}_{a,0}(u_d) > 0$ and $\bar{I}$ is convex in $u_a \in [0, \tilde{u}_{a,0}(u_d)]$ and concave in $u_a > \tilde{u}_{a,0}(u_d)$.
\end{proof}

The following corollary presents the analytical result when $\gamma$ and $\zeta$ are linear in $u_d$ and $u_a$, a special case of Prop.~\ref{prop:3}. 
\begin{corollary}
Let $\gamma(u_d) = k_d u_d + \gamma_0$ and $\zeta(u_a) = k_a u_a + \zeta_0$. We have a unique $\tilde{u}_{a,0}(u_d) = \frac{1}{2d_{\min} k_a}(k_d u_d + \gamma_0-2d_{\min}\zeta_0)$, which is always positive if $\zeta(0) < \frac{\gamma(0)}{2d_{\min}}$. Given any $u_d \in \mathcal{U}_d$, $\bar{I}(u_d, u_a)$ is convex in $u_a \in [0, \tilde{u}_{a,0}(u_d)]$ and concave in $u_a > \tilde{u}_{a,0}(u_d)$.
\end{corollary}

\begin{remark}
When $\bar{I}(u_d, u_a)$ is convex in $u_a$ given $u_d$, it indicates that the attacker has an increasing marginal effect in expanding the malware infection. The more attack effort, the faster the malware infection happens in the IoT device network. However, the rapid expansion is only valid within a threshold $u_a < \tilde{u}_{a,0}$. If the attacker keeps putting in more attack effort, the infection rate decreases. Rapid malware infection is undesired in practice, which is equivalent to reducing the region $ 0\leq u_a < \tilde{u}_{a,0}$. One way to achieve this is to improve the defender's defense effort $u_d$, which can decrease $\tilde{u}_{a,0}$. This fact can be directly observed in Corollary 3.1.
\end{remark}

\subsubsection{Characterization of Nash Equilibrium}
We first define the NE of the cyber defense game as follows.
\begin{definition}[Nash Equilibrium of Cyber Defense Game]
A strategy pair $(u_d^{\mathrm{NE}}, u_a^{\mathrm{NE}})$ constitutes a NE of the cyber defense game \eqref{eqn:cyber_game} if 
\begin{equation*}
    L_d(u_d^{\mathrm{NE}}, u_a^{\mathrm{NE}}) \leq L_d(u_d, u_d^{\mathrm{NE}}), \quad \forall u_d \in \mathcal{U}_d,
\end{equation*}
\begin{equation*}
    L_a(u_d^{\mathrm{NE}}, u_a^{\mathrm{NE}}) \geq L_a(u_d^{\mathrm{NE}}, u_a), \quad \forall u_a \in \mathcal{U}_a.
\end{equation*}
\end{definition}

We analyze the existence condition of the NE by adopting the convex-concave property of $\bar{I}$ discussed in Prop.~\ref{prop:2}. Note that the defender's objective function $L_d(u_d, u_a)$ is strictly convex in $u_d$ for any $u_a \in \mathcal{U}_a$ Therefore, the optimal response is unique, denoted by $BR_d(u_a) = \arg\min_{u_d} L_d(u_d, u_a)$, which can be computed by $\nabla_{u_d} L_d(u_d, u_a) = 0$. 
However, the convex-concave property breaks the uniqueness of the optimal response $BR_a(u_a) = \arg\min_{u_a} L_a(u_d, u_a)$, which can affect the search for NE.
We note that the convex-concave property indicates that $\frac{\partial \bar{I}}{\partial u_a}$ is unimodal given any $u_d \in \mathcal{U}_d$. i.e., $\frac{\partial \bar{I}}{\partial u_a}$ first increases in $u_a \in [0, \tilde{u}_{a,0})$ and then decreases in $u_a \geq \tilde{u}_{a,0}$. Then, we have the following proposition to characterize the existence of NE. 

\begin{proposition} \label{prop:4}
Suppose $\gamma$ and $\zeta$ are chosen to satisfy the conditions in Prop.~\ref{prop:2}. Then, the attacker's optimal response $BR_a(u_a)$ admits a unique value for every $u_d \in \mathcal{U}_d$, if $c_a'(u_a) \leq \frac{\partial \bar{I}}{\partial u_a}\big\vert_{u_d, u_a}$ for all $u_a \in [0, \tilde{u}_{a,0})$. Therefore, the cyber defense game admits an NE solution.
\end{proposition}

\begin{proof}
We can check that $\frac{\partial \bar{I}}{\partial u_a} = e^{-\frac{\gamma}{d\zeta}} \frac{\gamma \zeta'}{d\zeta^2}$ and $\frac{\partial \bar{I}}{\partial u_a} \to 0$ as $u_a\to\infty$ given any $u_d \in \mathcal{U}_d$. On the other hand, we have $\lim_{u_a \to \infty} C_a'(u_a) = \infty$ from assumptions $C_a' > 0$ and $C_a'' > 0$. Therefore, there is only one intersection of $C_a'(u_a)$ and $\frac{\partial \bar{I}}{\partial u_a}$ in the entire $\mathcal{U}_a$. It indicates that $L_a(u_d, u_a)$ only has one maximizer in $\mathcal{U}_a$ for any given $u_d \in \mathcal{U}_d$. Thus, $BR_a(u_d)$ is unique. 
Note that from the assumptions on $C_d$ and $C_a$, we obtain that $L_d(u_d, u_a) \to \infty$ as $u_d \to \infty$ for any given $u_a$, $L_a(u_d, u_a) \to \infty$ as $u_a \to \infty$ as any given $u_d$. Therefore, using the Brouwer's fixed point theorem, the equations
\begin{equation*}
    \begin{bmatrix} u_d \\ u_a \end{bmatrix} = 
    \begin{bmatrix} BR_d(u_a) \\ BR_a(u_d) \end{bmatrix}
\end{equation*}
admits a fixed point $(u^*_d, u^*_a)$, i.e., $L_d(u^*_d, u^*_a) \leq L_d(u_d, u^*_a)$ $\forall u_d \in \mathcal{U_d}$ and $L_a(u^*_d, u^*_a) \geq L_a(u^*_d, u_a)$ $\forall u_a \in \mathcal{U_a}$ $\forall u_a \in \mathcal{U}_a$. Thus, $(u^*_a, u^*_a)$ forms the NE of the cyber defense game.
\end{proof}

The condition in Prop.~\ref{prop:4} is mild. It can be easily satisfied by selecting proper $C_a$, especially when $u_d$ is large and the corresponding $\tilde{u}_{a,0}(u_d)$ is very small. We have the following corollary to discuss the linear-quadratic choices of $\gamma$, $\zeta$, and $C_a$.

\begin{corollary}
For linear choices of $\gamma$ and $\zeta$ and quadratic choice of $C_a = \frac{1}{2} c_a u_a^2$. The condition in Prop.~\ref{prop:4} to guarantee the unique $BR_a$ becomes $e^{-\frac{\gamma(u_d)}{d_{\min}\zeta_0}} \cdot \frac{\gamma(u_d) k_a}{d_{\min}\zeta_0} \geq c_a$ for every given $u_d$. Since $u_d = \gamma^{-1}(d_{\min}\zeta_0)$ maximizes the left-hand side term, the condition is further simplified to $e^{-1} k_a \geq c_a$.
\end{corollary}

Following Prop.~\ref{prop:4}, we develop the following iterative algorithm to find an NE of the cyber defense game.

\begin{algorithm}
\KwIn Degree distribution $p$, $\gamma, \zeta$, $\epsilon$ \;
$\langle k \rangle \gets \sum_k k p(k)$, $\langle k^2 \rangle \gets \sum_k k^2 p(k)$ \;
$i \gets 0$ \;
Initialize $u_{d,(0)}, u_{a,(0)}$ such that $\frac{\gamma(u_{d,(0)})}{\zeta(u_{a,(0)})} < \frac{\langle k^2 \rangle}{\langle k \rangle}$ \;
\While{$i < i_{\max}$}{
    Solve $BR_d(u_{a,(i)}) = \arg\min_{u_d} L_d(u_d, u_{a,(i)})$ using gradient descent \;
    Solve $BR_a(u_{d,(i)}) = \arg\min_{u_a} L_a(u_{d,(i)}, u_a)$ using gradient descent \;

    \uIf{$\| BR_d(u_{a,(i)}) - u_{d,(i)} \| < \epsilon$ \KwAnd $\| BR_a(u_{d,(i)}) - u_{a,(i)} \| < \epsilon$}{
        $u_d^{\mathrm{NE}} \gets BR_d(u_{a,(i)})$ \;
        $u_a^{\mathrm{NE}} \gets BR_a(u_{d,(i)})$ \;
        \KwBreak;
    }
    
    $u_{d, (i+1)} \gets BR_d(u_{a,(i)})$; \\
    $u_{a, (i+1)} \gets BR_d(u_{a,(i)})$ \;
    $i \gets i + 1$ \;
}
\KwOut NE of the cyber defense game $(u_d^{\mathrm{NE}}, u_a^{\mathrm{NE}})$.
\caption{Iterative method to find NE}
\label{alg:0}
\end{algorithm}

The NE of the cyber defense game provides a cyber-resilient strategy for the defender to counteract the botnet attack in IoT device networks. Under the equilibrium strategy, the malware infection reaches a steady state $\bar{I}(u_d^{\mathrm{NE}}, u_a^{\mathrm{NE}})$, which leads to a systemic risk $\bar{R}=\bar{I}(u_d^{\mathrm{NE}}, u_a^{\mathrm{NE}}) \cdot N_d \cdot W_d$. The systemic risk $\bar{R}$ is the maximum amount of vulnerable loads that the attacker can control in the power grid, which provides the upper bound for the attacker to implement malicious load manipulation at the physical layer. The system operator then needs to develop resilient grid regulation strategies by considering the consequences of the cyber layer interactions, which will be pursued in the next section.

\section{Dynamic Game for Physical Resilience}\label{sec:physical_design}
The systemic risk $\bar{R}$ at the steady state denotes an aggregated level of IoT-controlled energy devices that can be compromised. To quantify the individual risk of each bus, we first need to know how these IoT-controlled energy devices are distributed over the power network. 

Recall that $\mathcal{L} = \{l_1,\dots,l_{N_L} \}$ denotes the set of load buses in the grid. Let $\rho \in [0,1]^{N_L}$ be the distribution of IoT-controlled energy devices, where $\sum_{i=1}^{N_L} \rho_i = 1$. Then, $\rho_i$ indicates the proportion of the IoT-operated devices presented at bus $i$, and the corresponding risk of bus $i$ can be expressed by $\bar{R}\cdot \rho_i := P^{LV}_i$. For example, when these vulnerable IoT-controlled energy devices are uniformly distributed in the network, each bus faces the same level of risks $\bar{R}/N_L$. This quantification indicates the maximum load the adversary at each bus can control. 

After identifying how much load can be altered, the attacker determines its attack strategy $\mathbf{P}^a$ to destabilize the grid, which could result in generator shutdown and even cascading failures. $P^a_i(t)$ the amount of load changed by the attacker at bus $i$ at time $t$, $i\in \mathcal{L}$. The attack strategy needs to consider the feasibility constraint, i.e., $P^a_i(t)\leq \bar{R}\cdot \rho_i$, during the execution of the attack. 

We use discrete system \eqref{eqn:dyn_3} to achieve grid control.
The system operator regulates the input power $\{\mathbf{P}^d_t\}_{t\in\mathbb{Z}_+}$ to counteract malicious load manipulation by anticipating the strategic attack behavior. Therefore, the operator can decide the regulation strategy to stabilize the system using the following min-max controller:
\begin{equation}
\label{eqn:minmax}
\begin{split}
    \min_{\{ \mathbf{P}^d \}} \max_{\{\mathbf{P}^a\}} \quad & J \left( \{\mathbf{P}^d \}, \{\mathbf{P}^a\} \right) := \| \mathbf{x}_T \|^2_{Q_f} \\
    & \quad+ \sum_{t=0}^{T-1} \left( \| \mathbf{x}_t \|^2_{Q} + \| \mathbf{P}^d_t \|^2_{R_d} - \| \mathbf{P}^a_t \|^2_{R_a} \right) \\ 
    \text{s.t.} \quad & \mathbf{x}_{t+1} = \tilde{A} \mathbf{x}_t + \tilde{B}_d \mathbf{P}^d_t + \tilde{B}_a \mathbf{P}^a_t + \tilde{c},  \\ 
    & P^a_{t,i} \leq \bar{R} \cdot \rho, \quad \forall i \in \mathcal{L}, \ \forall t.
\end{split}
\end{equation}
Here, the state $\mathbf{x}_t$ and the controls $\mathbf{P}^d_t, \mathbf{P}^a_t$ align with the discrete dynamical system \eqref{eqn:dyn_3}. $Q \succeq 0, Q_f \succeq 0, R_d \succ 0, R_a \succ 0$ are state and control cost matrices with proper dimensions. For simplicity, we denote $\{ \mathbf{P}^d \} := \{ \mathbf{P}^d_t \}_{t=0}^{T-1}$ and $\{ \mathbf{P}^a \} := \{ \mathbf{P}^a_t \}_{t=0}^{T-1}$. 
The defender aims to stabilize the system using additional regulation power $\{\mathbf{P}^d\}$, while the attacker wants to destabilize the system by manipulating the vulnerable load $\{ \mathbf{P}^a \}$. The interactions between the defender and attacker at the physical layer naturally constitute a dynamic game. Note that the attacker can both increase and decrease the vulnerable loads within the range of the feasibility constraints since the IoT devices may be in use when compromised.

\begin{remark}
The diagonal weighting matrices, $R_d$ and $R_a$, quantify the resource costs associated with grid regulation and conducting an attack. The value of each diagonal element reflects the level of effort to generate power or manipulate the load at the corresponding bus.
\end{remark}

The attacker's feasibility constraints raise challenges in solving the problem \eqref{eqn:minmax}. We use log barrier functions to penalize the constraint violation and add them to the objective function. Then, the modified objective becomes
\begin{equation*}
    \tilde{J} \left( \{\mathbf{P}^d \}, \{\mathbf{P}^a\} \right) = J+ \sum_{t=0}^{T-1} \sum_{i=1}^{N_L} \frac{1}{\mu} \log (\bar{R}\cdot \rho_i - P^a_{t,i}).
\end{equation*}
The defender can instead use a modified controller resulting from the following problem to regulate the grid under the IoT botnet attack:
\begin{equation}
\label{eqn:minmax_1}
\begin{split}
    \min_{\{ \mathbf{P}^d \}} \max_{\{\mathbf{P}^a\}} \quad & \tilde{J} \left( \{\mathbf{P}^d \}, \{\mathbf{P}^a\} \right) \\ 
    \text{s.t.} \quad & \mathbf{x}_{t+1} = \tilde{A} \mathbf{x}_t + \tilde{B}_d \mathbf{P}^M_t + \tilde{B}_a \mathbf{P}^a_t + \tilde{c}.
\end{split}
\end{equation}

The min-max control formulation in \eqref{eqn:minmax_1} defines a dynamic Nash game \cite{bacsar1998dynamic} at the physical layer, which we call the physical dynamic game. We use open-loop NE as the solution concept to solve \eqref{eqn:minmax_1}. 

\begin{definition}[Open-loop Nash Equilibrium of Physical Dynamic Game]
A strategy trajectory pair $(\{ \mathbf{P}^{d*} \}, \{ \mathbf{P}^{a*} \})$ constitutes an open-loop NE of the dynamic game \eqref{eqn:minmax_1} if
\begin{equation*}
    \tilde{J}(\{ \mathbf{P}^{d*}\}, \{ \mathbf{P}^a\}) \leq \tilde{J}(\{ \mathbf{P}^{d*}\}, \{ \mathbf{P}^{a*}\}) \leq \tilde{J}(\{ \mathbf{P}^d\}, \{ \mathbf{P}^{a*}\}) 
\end{equation*}
for all $\{ \mathbf{P}^d\}$ and $\{ \mathbf{P}^a\}$ with $\mathbf{P}^a_t \leq \bar{R} \cdot \rho$, $t=0,\dots, T-1$.
\end{definition}

The following proposition specifies the existence condition of the equilibrium solution to the physical dynamic game.

\begin{proposition} \label{prop:5}
The min-max problem \eqref{eqn:minmax_1} admits a unique open loop NE if $R_d + \tilde{B}_d^\tp S^d_{t+1} \tilde{B}_d \succ 0$ for $t=0,\dots, T-1$, where $S^d$ is updated by 
\begin{equation*}
\begin{split}
    S^d_{t} &= Q + \tilde{A}^\tp S^d_{t+1} \tilde{A} \\
    &- \tilde{A}^\tp S^d_{t+1} \tilde{B}_d (R_d + \tilde{B}_d^\tp S^d_{t+1} \tilde{B}_d)^{-1} \tilde{B}_d^\tp S^d_{t+1} \tilde{A}
\end{split}
\end{equation*}
with $S^d_T = Q_f$, and $R_a - \tilde{B}_a^\tp S^a_{t+1} \tilde{B}_a \succ 0$ for $t=0,\dots, T-1$, where $S^a$ is updated by
\begin{equation*}
\begin{split}
    S^a_{t} &= Q + \tilde{A}^\tp S^a_{t+1} \tilde{A} \\
    &+ \tilde{A}^\tp S^a_{t+1} \tilde{B}_a (R_a - \tilde{B}_a^\tp S^a_{t+1} \tilde{B}_a)^{-1} \tilde{B}_a^\tp S^a_{t+1} \tilde{A}
\end{split}
\end{equation*}
with $S^a_T = Q_f$.
\end{proposition}

\begin{proof}
We first show the sufficient conditions for the objective function $\tilde{J}$ to be convex in $\{ \mathbf{P}^d \}$ and concave in $\{ \mathbf{P}^a \}$.
Note that the objective function $\tilde{J}$ is quadratic in $\{ \mathbf{P}^d \}$. We want to show for any $\{ \mathbf{P}^A \}$, $\tilde{J}$ is strictly convex in $\{ \mathbf{P}^d \}$, which is equivalent to the following optimal control problem. The Hessian of $\tilde{J}$ is not related to the constant term $\tilde{c}$ and $\{ \mathbf{P}^a \}$. So it is equivalent to setting them to 0. We obtain 
\begin{equation*}
\begin{split}
    \min_{\{\mathbf{P}^d\}} \quad & \| \mathbf{x}_T \|^2_{Q_f} + \sum_{t=0}^{T-1} \left( \| \mathbf{x}_t \|^2_{Q} + \| \mathbf{P}^d_t \|^2_{R_d} \right) \\ 
    \text{s.t.} \quad & \mathbf{x}_{t+1} = \tilde{A} \mathbf{x}_t + \tilde{B}_d \mathbf{P}^d_t.
\end{split}
\end{equation*}
The optimal control problem has a unique solution if and only if the Riccati equation holds, which is 
\begin{equation*}
\begin{split}
    S^d_{t} &= Q + \tilde{A}^\tp S^d_{t+1} \tilde{A} \\
    &- \tilde{A}^\tp S^d_{t+1} \tilde{B}_d (R_d + \tilde{B}_d^\tp S^d_{t+1} \tilde{B}_d)^{-1} \tilde{B}_d^\tp S^d_{t+1} \tilde{A}
\end{split}
\end{equation*}
with $S^d_T = Q_f$, and $R_d + \tilde{B}_d^\tp S^d_{t+1} \tilde{B}_d \succ 0$, $t=0,\dots,T-1$. The latter guarantees the convexity of the value function. Otherwise, it becomes unbounded.

Likewise, we show that $\tilde{J}$ is concave in $\{ \mathbf{P}^a \}$ for any $\{ \mathbf{P}^d \}$. We note that besides the quadratic terms, the log barrier terms only contain $\{ \mathbf{P}^a \}$ and are readily concave in $\{ \mathbf{P}^a \}$. Therefore, it is sufficient to ignore the log functions and only consider quadratic terms in $\{\mathbf{P}^a\}$. Following the same arguments as $\{\mathbf{P}^d\}$, we aim to show that the following optimal control problem has a unique solution:
\begin{equation*}
\begin{split}
    \max_{\{\mathbf{P}^a\}} \quad & \| \mathbf{x}_T \|^2_{Q_f} + \sum_{t=0}^{T-1} \left( \| \mathbf{x}_t \|^2_{Q} - \| \mathbf{P}^a_t \|^2_{R_a} \right) \\ 
    \text{s.t.} \quad & \mathbf{x}_{t+1} = \tilde{A} \mathbf{x}_t + \tilde{B}_a \mathbf{P}^a_t,
\end{split}
\end{equation*}
which provides the Riccati equations 
\begin{equation*}
\begin{split}
    S^a_{t} &= Q + \tilde{A}^\tp S^a_{t+1} \tilde{A} \\
    &+ \tilde{A}^\tp S^a_{t+1} \tilde{B}_a (R_a - \tilde{B}_a^\tp S^a_{t+1} \tilde{B}_a)^{-1} \tilde{B}_a^\tp S^a_{t+1} \tilde{A}
\end{split}
\end{equation*}
with $S^a_T = Q_f$, and $R_a - \tilde{B}_a^\tp S^a_{t+1} \tilde{B}_a \succ 0$, $t=0,\dots,T-1$.

Besides, $\tilde{J}\to \infty$ as $|\{\mathbf{P}^d\}| \to \infty$ for any $\{\mathbf{P}^a\}$ and $\tilde{J}\to -\infty$ as $|\{\mathbf{P}^a\}| \to \infty$ for any $\{\mathbf{P}^d\}$. Using the minimax theorem, the game \eqref{eqn:minmax_1} admits a unique open-loop NE $(\{\mathbf{P}^{d*}\}, \{\mathbf{P}^{a*}\})$.
\end{proof}

From Prop.~\ref{prop:5}, the state cost matrix $Q$ does not have to be positive definite or negative definite. As long as the conditions are satisfied, the Nash game \eqref{eqn:minmax_1} admits an open-loop NE.
Next, we find the necessary conditions to compute the open-loop NE of \eqref{eqn:minmax_1}. 

\begin{proposition}
Suppose the convex-concave condition in Prop.~\ref{prop:5} is satisfied. Further assume $R_a \succ 0$ and $R_d \succ 0$ are diagonal matrices with element $r_{a,i}$ and $r_{d,i}$. Then, $\{ \mathbf{P}^{d*}\}$ and $\{ \mathbf{P}^{a*} \}$ is an open-loop NE of \eqref{eqn:minmax_1} if and only if there exists a trajectory $\{ \mathbf{x}_t \}_{t=1}^T$ and $\{ \lambda_t \}_{t=1}^T$ such that 
\begin{equation*}
\begin{split}
    &\mathbf{x}_{t+1} = \tilde{A} \mathbf{x}_t + \tilde{B}_d \mathbf{P}^{d*}_t + \tilde{B}_a \mathbf{P}^{a*}_t + \tilde{c}, \quad \mathbf{x}_0 \text{given},\\
    &\lambda_t = \tilde{A}^\tp \lambda_{t+1} + 2Q \mathbf{x}_t, \quad \lambda_T = 2Q_f \mathbf{x}_T,\\
    &\mathbf{P}^{d*}_t = -\frac{1}{2} R_d^{-1} \tilde{B}_d^\tp \lambda_{t+1}, \\
    &P^{a*}_{t,i} = \frac{2r_{i} \bar{R}\rho_i + [B_a^\tp \lambda_{t+1}]_i}{4 r_i} - \frac{1}{4r_i} \Big[ (2r_{i} \bar{R}\rho_i + [B_a^\tp \lambda_{t+1}]_i)^2  \\
    & \hspace{1.5cm} + 8r_i ([B_a^\tp \lambda_{t+1}]_i \bar{R}\rho_i - 1/\mu) \Big]^{1/2}, \ i=1,\dots, N_L.
\end{split}
\end{equation*}
\end{proposition}

\begin{proof}
The necessary condition for $(\{ \mathbf{P}^{d*}\}, \{ \mathbf{P}^{a*} \})$ being an open-loop NE is that there exist trajectories $\{ \mathbf{x}_t \}_{t=0}^{T}$ and $\{ \lambda_t \}_{t=1}^T$ such that
\begin{equation}
\label{eqn:necessary_cond}
\begin{split}
    &\mathbf{x}_{t+1} = \tilde{A} \mathbf{x}_t + \tilde{B}_d \mathbf{P}^{d*}_t + \tilde{B}_a \mathbf{P}^{a*}_t + \tilde{c}, \quad \mathbf{x}_0 \text{ given}, \\
    &H_t(\mathbf{x}_t, \mathbf{P}^{d*}_t, \mathbf{P}^a_t, \lambda_{t+1}) \leq H_t(\mathbf{x}_t, \mathbf{P}^{d*}_t, \mathbf{P}^{a*}_t, \lambda_{t+1}) \\
    & \qquad \leq H_t(\mathbf{x}_t, \mathbf{P}^d_t, \mathbf{P}^{a*}_t, \lambda_{t+1}), \quad \forall \mathbf{P}^d_t, \ \forall \mathbf{P}^a_t, \\ 
    &\lambda_t = \nabla_x H_t(\mathbf{x}_t, \mathbf{P}^{d*}_t, \mathbf{P}^{a*}_t, \lambda_{t+1}), \quad \lambda_T = 2Q_f \mathbf{x}_T,
\end{split}
\end{equation}
where the Hamiltonian $H_t$, $t=0,\dots, T-1$, is given by
\begin{equation*}
\begin{split}
    H_t(\mathbf{x}_t, \mathbf{P}^{d}_t, \mathbf{P}^{a}_t, \lambda_{t+1}) = \lambda_{t+1}^\tp \left( \tilde{A} \mathbf{x}_t + \tilde{B}_d \mathbf{P}^d_t + \tilde{B}_a \mathbf{P}^a_t + \tilde{c} \right) \\
    + \| \mathbf{x}_t \|^2_{Q} + \| \mathbf{P}^d_t \|^2_{R_d} - \| \mathbf{P}^a_t \|^2_{R_a} + \sum_{i=1}^{N_L} \frac{1}{\mu} \log (\bar{R}\cdot \rho_i - P^a_{t,i}).
\end{split}
\end{equation*}
We note that the $H_t$ is convex in $\mathbf{P}^d_t$ and concave in $\mathbf{P}^d_t$. Then we can use the first order condition to simplify the second inequality in \eqref{eqn:necessary_cond} and obtain
\begin{equation*}
    \mathbf{P}^{d*}_t = -\frac{1}{2} R_d^{-1} B_d^\tp \lambda_{t+1}.
\end{equation*}
Using the assumption that $R_a$ is diagonal, we check the first inequality elementwise that 
\begin{equation*}
    2 r_i (P^a_{t,i})^2 - (2 r_i \bar{R}\rho_i + [\tilde{B}_a^\tp \lambda_{t+1}]_i) P^a_{t,i} +\bar{R}\rho_i [\tilde{B}_a^\tp \lambda_{t+1}]_i - \frac{1}{\mu} = 0,
\end{equation*}
where $[\tilde{B}_a^\tp \lambda_{t+1}]_i$ represents the $i$-th element of $\tilde{B}_a^\tp \lambda_{t+1}$. Since $P^a_{t,i} \leq \bar{R} \cdot \rho_i$, we only take the negative root and arrive at the result in the proposition.
Since the convex-concave condition is also satisfied, the necessary conditions are also sufficient, which means that the solution $\{ \mathbf{P}^{d*}\}$ and $\{ \mathbf{P}^{a*} \}$ is unique if the conditions \eqref{eqn:necessary_cond} are met.
\end{proof}

Therefore, we can solve a feasibility problem with necessary conditions \eqref{eqn:necessary_cond} to find the unique open-loop NE of the game \eqref{eqn:minmax_1}. To facilitate computation, we can first solve the following unconstrained dynamic Nash game 
\begin{equation*}
\begin{split}
    \min_{\{ \mathbf{P}^d \}} \max_{\{\mathbf{P}^a\}} \quad & J \left( \{\mathbf{P}^d \}, \{\mathbf{P}^a\} \right) \\ 
    \text{s.t.} \quad & \mathbf{x}_{t+1} = \tilde{A} \mathbf{x}_t + \tilde{B}_d \mathbf{P}^M_t + \tilde{B}_a \mathbf{P}^a_t + \tilde{c},
\end{split}
\end{equation*}
whose necessary conditions are linear equations. Then, we use the solution as the initial guess to iteratively refine the solution of \eqref{eqn:necessary_cond} by increasing the barrier parameter $\mu$. The refined solution gradually approaches the open-loop NE of \eqref{eqn:minmax}. We summarize the procedure in Alg.~\ref{alg:1} below.

\begin{algorithm}
\KwIn Initial state $x_0$ \;
$\alpha \gets 5$, $\mu \gets 2$ \;
$\{\mathbf{P}^d\}^{(0)}, \{\mathbf{P}^a\}^{(0)} \gets$ solve an unconstrained min-max controller with $x_0$ as an initial guess\;
$n \gets 0$ \;
\While{$n < n_{\max}$ }{
    $\{\mathbf{P}^d\}^{(n+1)}$, $\{\mathbf{P}^a\}^{(n+1)}$ $\gets$ solving \eqref{eqn:minmax_1} with initial conditions $\{\mathbf{P}^d\}^{(n)}$ and $\{\mathbf{P}^a\}^{(n)}$ and $x_0$ \;
    $\mu \gets \alpha \mu$ \;
    $n \gets n+1$ \;
}
\KwOut Approximate NE $(\{\mathbf{P}^d\}^{(n)}, \{\mathbf{P}^a\}^{(n)})$.
\caption{Iterative refinement of the open-loop NE.}
\label{alg:1}
\end{algorithm}

\begin{remark}
In the Nash games \eqref{eqn:minmax} and \eqref{eqn:minmax_1}, we assume that the attacker has full knowledge of the grid operator, including the model and defensive strategy. In practice, the attacker can utilize either publicly available grid data to estimate the operator's parameters or adversarial approaches like social engineering to obtain critical information. Recent advances in learning-based techniques (e.g., \cite{wang2020multi,xie2020imitation,wang2021review}) also offer possibilities for parameter estimation.
\end{remark}

Since the defender uses the open-loop NE to strategically regulate the grid to counteract the malicious load manipulation, we develop a receding horizon planning algorithm for the defender to generate successive power regulation strategies as time evolves, which is summarized in Alg.~\ref{alg:2}.

\begin{algorithm}
\While{botnet attack detection}{
\tcp{botnet attack detection loop, periodic detect or set manually}
        Run cyber defense and get systematic risk $\bar{R}$ \;
        Identify current state $x_0$ \;
        Control time step $t \gets 0$ \;
    \While{No new detection}{
        $(\{\mathbf{P}^{d*}\}, \{ \mathbf{P}^{a*}\}) \gets$ Run Alg.~\ref{alg:1} with $x_t$ \;
        Defender chooses $\mathbf{P}^{d*}_0$ \;
        Attacker decides $\mathbf{P}^a \leq \bar{R}\cdot\rho$ for load manipulation, may not play NE strategy $\mathbf{P}^{a*}_0$ \;
        $x_{t+1} \gets \tilde{A} \mathbf{x}_t + \tilde{B}_d \mathbf{P}^{d*}_0 + \tilde{B}_a \mathbf{P}^a + \tilde{c} $ from \eqref{eqn:dyn_3} \;
        $t \gets t+1$ \;
    }
}
\caption{Receding horizon planning for successive resilient control.}
\label{alg:2}
\end{algorithm}

\begin{remark}
In practice, the grid operator can employ Supervisory Control and Data Acquisition (SCADA) systems to monitor critical parameters such as power flow, load variations, and frequency levels across different buses. By examining factors like sudden load shifts or comparing them with historical data, the grid operator can detect potential load altering attacks. With this awareness, the operator can then formulate a targeted defense strategy to stabilize the grid.
\end{remark}

\section{Case Studies} \label{sec:case}

We use the IEEE-39 bus system (10 generator buses and 29 load buses) to showcase the results. The system parameters, including the transmission lines and the inertia and damping coefficients of generators, are the same as those provided in MATPOWER simulator\footnote{\url{https://matpower.org/docs/ref/matpower5.0/case39.html}}. The nominal system frequency is $\omega_{\text{n}} = 60$ Hz and the maximum deviation is $\omega_{\text{max}} = 2$ Hz, i.e., the generator's relay trips at 62 Hz (over-frequency) and 58 Hz (under-frequency). The grid dynamics are simulated according to \eqref{eqn:dyn_3}.
The set of vulnerable buses is $\mathcal{V}=\{6, 10,12,15,16, 19, 23, 29\}$, where massive IoT-controlled high-power devices are present. We assume the uniform distribution of IoT devices. 
The damping coefficient of each load is 10. The coefficients of the PI controller associated with generator buses are: ${\bf K}^P= \mathrm{diag}([20, 15, 15, 10, 10, 10, 30, 10, 25, 10])$, and $\mathbf{K}^I = \mathrm{diag}([45,35,45,60,35,58,30,30,50,50])$, where $\mathrm{diag}$ denotes the diagonal operator.
The per unit (p.u.) value of power is 100 MW.
The parameters associated with the cyber layer are below: $d_{\min} = 1$, $N_d = 10^7$, $W_d = 5,000$W.\footnote{We set the average power $W_d=5000$ to resemble IoT-enabled energy devices like cloth dryers and air conditioners. The total device number $N_d$ is based on the estimation of energy-related smart home appliances in NYC \cite{statista2023number,mariotti2023smart}.}
Thus, the maximum load of all the IoT high-power devices connected to the vulnerable buses is 500 p.u. The average vulnerable load in each bus in $\mathcal{V}$ is $62.5$ p.u. We adopt the secure load data from \cite{amini2016dynamic} and set secure load as $\mathbf{P}^{LS} = 171.7$ p.u.

\begin{remark}
The vulnerable load $\mathbf{P}^{LV}$ is the total load of IoT devices, which represents the upper limit of load manipulation available to an attacker. The attacker can only alter a portion of 500 p.u. unless he compromises every IoT device in the power grid. However, infecting all devices is infeasible because the defender also has cyber defense strategies to protect the IoT devices.
\end{remark}

The attacker's manipulation of IoT energy devices is dynamic. In the following case studies, we assume that the attacker can conduct a coordinated attack, i.e., the attacker can maliciously manipulate loads in all vulnerable load buses $\mathcal{V}$ simultaneously.

\subsection{Cyber Risk and Cyber Defense Game Assessment} \label{sec:case.cyber}
We first assess the cyber risk $I(t)$ of the power grid under an IoT botnet attack. Specifically, we consider $\gamma = 0.2$ and $\zeta = [0.2, 0.25, 0.3,0.4, 0.5]$ and simulate $I(t)$ in Fig.~\ref{fig:cyber_risk}. Here, $\gamma$ and $\zeta$ represent the equivalent cyber protection and attack capabilities. As depicted in Fig.~\ref{fig:cyber_risk.1}, the cyber risk $I(t)$ grows as $\zeta$ increases. It shows that the attacker can manipulate more energy devices and thus has additional flexibility in devising the IoT botnet attack with a larger $\zeta$.
Fig.~\ref{fig:cyber_risk.2} depicts the cyber risk $\bar{I}$ at the steady state, which increases with $\zeta$ for a fixed $\gamma$. 
We also plot the approximation function $e^{-\gamma/d_{\min}\zeta}$ of $\bar{I}$ in \eqref{eqn:I_equ_final}. The comparison with the accurate $\bar{I}$ shows that the approximation performance is satisfactory.

\begin{figure}
    \centering
    \subfigure[$I(t)$ evolution for different $\zeta$.]{
        \includegraphics[width=0.45\columnwidth]{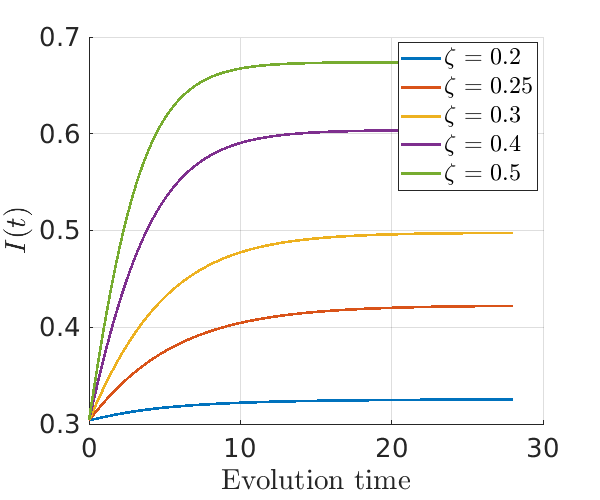}
        \label{fig:cyber_risk.1}
    }
    \subfigure[$\bar{I}$ and its approximation.]{
        \includegraphics[width=0.45\columnwidth]{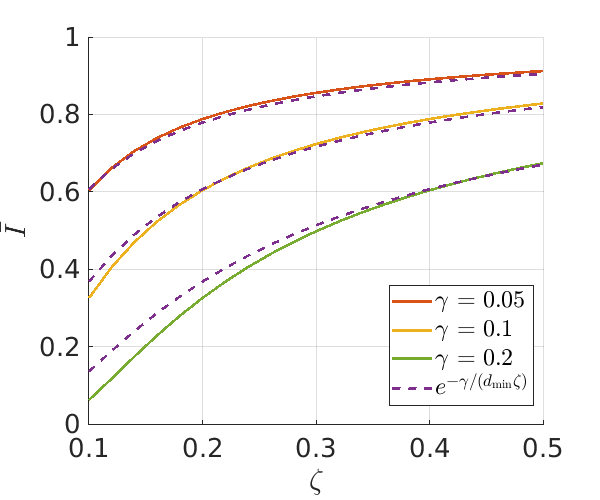}
        \label{fig:cyber_risk.2}
    }    
    \caption[]{(a) illustrates the percentage of the compromised IoT-controlled energy devices in the grid under the IoT botnet attacks with different attack intensities $\zeta$. (b) shows the resulting cyber risk $\bar{I}$ at the steady state as a function of $\zeta$ for fixed $\gamma$. The approximate function in \eqref{eqn:I_equ_final} yields satisfactory results for $\bar{I}$.}
  \label{fig:cyber_risk}
\end{figure}

In the cyber defense game, we set $\gamma(u_d) = \sqrt{u_d} + 0.1$ and $\zeta(u_a) = 2.5\log(u_a+1) + 0.1$ to represent the cyber defense/attack intensity by using the defense/attack effort $u_d$ and $u_a$. Note the choice of $\gamma$ and $\zeta$ are both concave in $u_d$ and $u_a$, which satisfy the conditions in Prop.~\ref{prop:3}. The cost functions are set as $C_d(u_d) = 0.2 u_d^2$ and $C_a(u_a) = 0.2 u_a^2$. 
We show two functions related to the attacker in Fig.~\ref{fig:cyber_game_ua} for detailed discussion. From Fig.~\ref{fig:cyber_game_ua.1} we observe that $\bar{I}(u_d, u_a)$ exhibits a convex-concave property in $u_a$ when $u_d$ is large (e.g., $u_d=1$ and $2$), which corroborates the results in Prop.~\ref{prop:3}. As $u_d$ reduces, $\bar{I}(u_d, u_a)$ becomes concave in $u_a$ because the convex-concave conditions are easier to meet. 
%
Following Prop.~\ref{prop:4}, Fig.~\ref{fig:cyber_game_ua.2} shows a unique maximizer $u_a^*$ of attacker's utility $L_a(u_d, u_a)$ for a fixed $u_d$, which is the attacker's optimal cyberattack strategy. Besides, $u^*_a$ increases as $u_d$ goes up, showing that the attacker has to put more attack effort into compromising more IoT devices if the defender escalates the cyber defense.

\begin{figure}[!t]
  \centering
  \subfigure[$\bar{I}$ changes with $u_a$.]{
    \includegraphics[width=0.48\columnwidth]{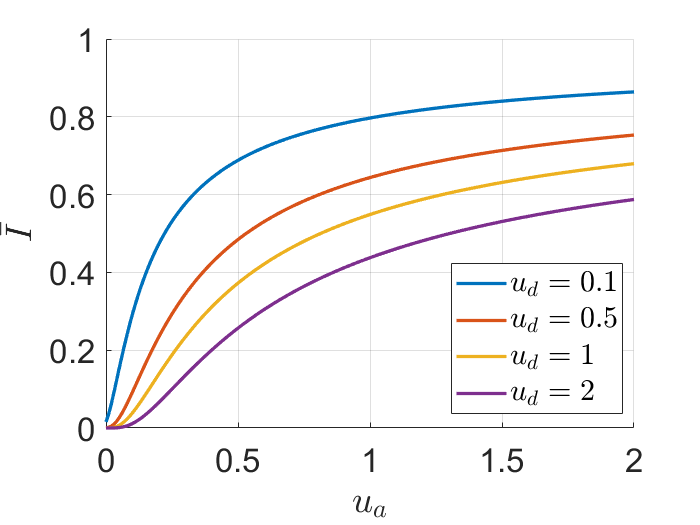}\label{fig:cyber_game_ua.1}}
	 \subfigure[$L_a$ changes with $u_a$.]{
    \includegraphics[width=0.48\columnwidth]{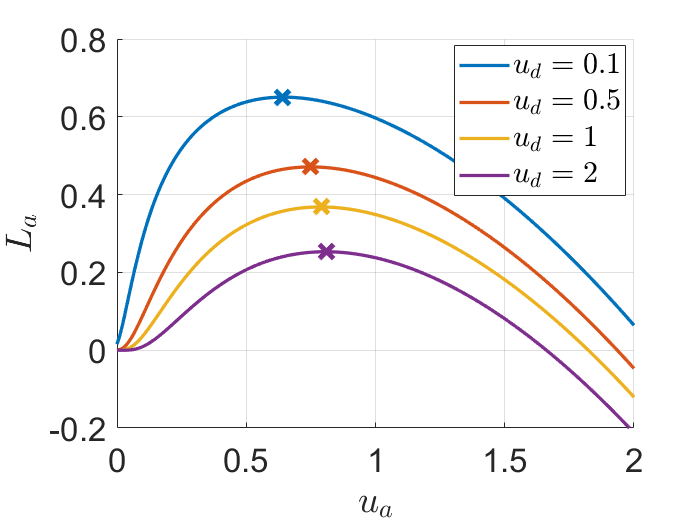}\label{fig:cyber_game_ua.2}}    
  \caption[]{(a) plots $\bar{I}(u_d,u_a)$ as a function of $u_a$ for fixed $u_d$. $\bar{I}$ behaves convex-concave for large $u_d$ and becomes concave as $u_d$ reduces. (b) shows that the attacker's utility $L_a(u_d, u_a)$ admits a unique maximizer $u_a^*$ for fixed $u_d$. $u_a^*$ increases with $u_d$, indicating that more attack effort is required when the defender escalates the cyber defense.}
  \label{fig:cyber_game_ua}
\end{figure}

To find the NE of the cyber defense game, the defender and attacker respond to each other's action optimally and repeatedly, as shown in Fig.~\ref{fig:cyber_game_nash}. The defender and attacker's optimal response functions are depicted in Fig.~\ref{fig:cyber_game_nash.1}. The intersection of two functions shows that the game admits a NE $(u_d^{\mathrm{NE}}, u_a^{\mathrm{NE}}) = (0.58, 0.76)$, which specifies the defense and attack strategies at the cyber layer.
Fig.~\ref{fig:cyber_game_nash.2} shows that Alg.~\ref{alg:0} successfully converges to the NE in Fig.~\ref{fig:cyber_game_nash.1}. The max difference is measured by the difference of defender/attacker's cyber action in two adjacent iterations in Alg.~\ref{alg:0}, i.e., $\max\{ |u_{d,(i+1)} - u_{d,(i)}|, |u_{a,(i+1)} - u_{a,(i)}| \}$, where $i$ represents the $i$-th iteration. 
The corresponding cyber risk admits $\bar{I} = 0.56$, which gives a vulnerable load $P^{LV}_i = \bar{R} \cdot \rho_i = 21$ p.u. for $i \in \mathcal{V}$. It provides an upper bound for the attacker's action during implementing the malicious load manipulation at the physical layer.

\begin{figure}[!t]
    \centering
    \subfigure[]{\includegraphics[width=0.48\columnwidth]{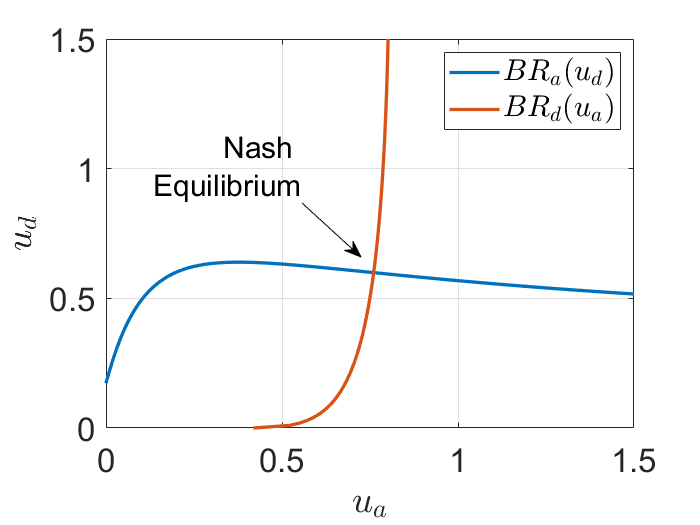}\label{fig:cyber_game_nash.1}}
	\subfigure[]{\includegraphics[width=0.48\columnwidth]{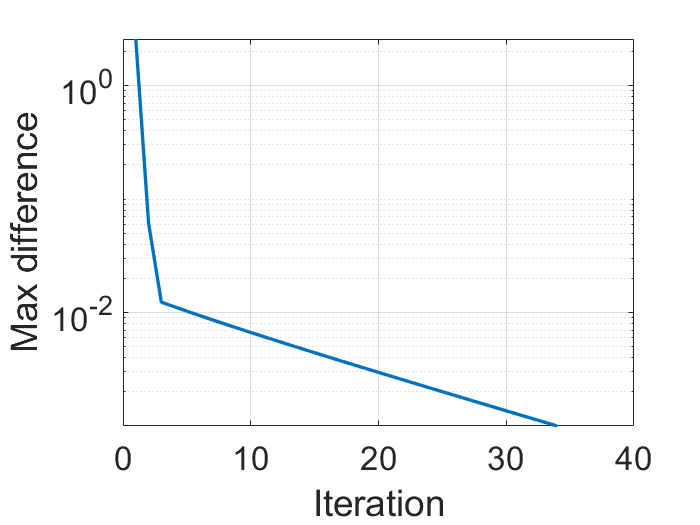}\label{fig:cyber_game_nash.2}}  
\caption[]{(a) depicts the defender and the attacker's optimal response functions. The intersection is the NE. (b) shows that Alg.~\ref{alg:0} can quickly converge to a NE. The max difference is measured by $\max\{ | u_{d,{i+1}}-u_{d,(i)}|, |u_{a,(i+1)}-u_{a,(i)}| \}$.}
\label{fig:cyber_game_nash}
\end{figure}

\subsection{Physical Impact of IoT Botnet Attack} \label{sec:case.physical}

Additional power regulation is critical for the defender to combat botnet attacks.
As depicted in Fig.~\ref{fig:nodefense.1}, the defender can use a pre-designed PI controller to stabilize the generator frequency in the permissible range when there are no attacks. The zoom-in plot shows that all generators' frequencies converge to $\omega_\mathrm{n}$ after some oscillations. However, the PI controller is insufficient to stabilize normal operation when the attacker maliciously manipulates the load.
We consider a load-switching attack, where the attacker turns on $0.9 \cdot P^{LV}_i=152$ p.u. for all $i \in \mathcal{V}$ in $0$-$50$s, then turns off the loads in $50$-$100$s, and again turns on $0.9 \cdot P^{LV}_i$ in $100$-$150$s. The defender only uses a PI controller to regulate the system. As shown in Fig.~\ref{fig:nodefense.2}, some generators' frequencies exceed the maximum permissible frequency deviation range, disrupting the system operation. This shows the necessity of using additional power regulation approaches to improve the grid's physical resilience.

\begin{figure}[!t]
    \centering
    \subfigure[zero attack.]{\includegraphics[width=0.48\columnwidth]{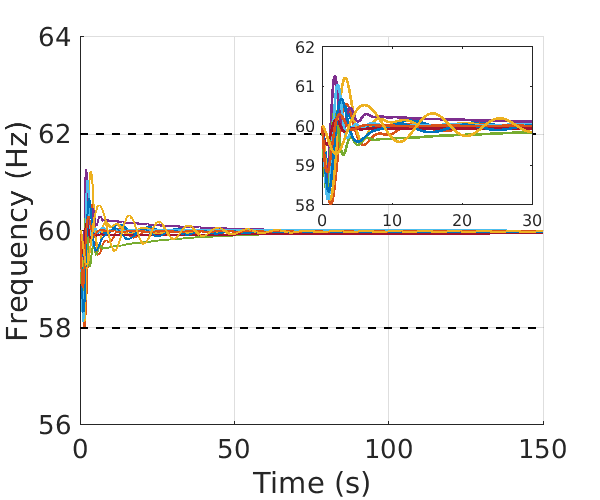}\label{fig:nodefense.1}}
	\subfigure[Load switching attack.]{\includegraphics[width=0.48\columnwidth]{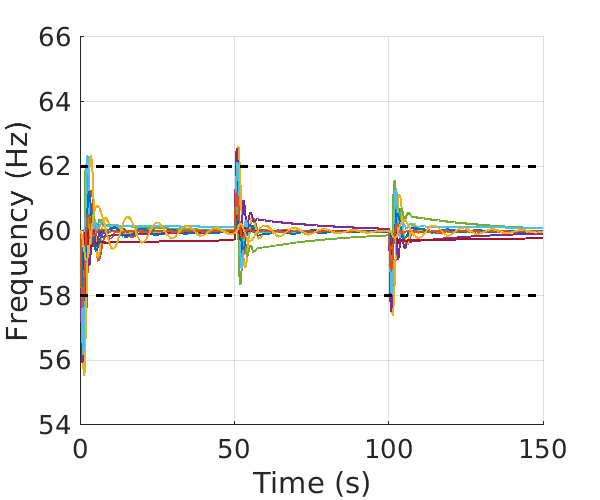}\label{fig:nodefense.2}}  
\caption[]{Generator frequency evolution using a pre-designed PI controller. (a) shows that the PI controller can stabilize the system when there is no attack. (b) implements a load-switching attack and generator frequencies are out of the permissible range, showing a single PI controller is insufficient to stabilize the system under the botnet attack, and it calls for additional resilience enhancement schemes.}
\label{fig:nodefense}
\end{figure}

To this end, the defender at the physical layer uses the min-max controller \eqref{eqn:minmax_1} and receding horizon planning to enable strategic defense against botnet attacks, as developed in Section \ref{sec:physical_design}. We set $Q = \mathrm{diag}([I_{10}, I_{29}, 5\cdot I_{10}])$ and put more penalty weights to stabilize the frequency deviation. We set $Q_f = 5 Q_f$. We set $R_{a,i} = 0.05$ for $i=1,\dots, N_V$ to capture the attacker's low cost of manipulating the load after compromising IoT devices. We set by $R_{d,i} = 0.2$ for $i=1,\dots, N_G$.

We simulate two attack scenarios. In the first scenario, the attacker uses the strategic attack generated from the min-max controller, which can be viewed as the worst-case attack. In the second one, the attacker manipulates a constant load $\mathbf{P}^a = [10.4, 10.6, 9.9, 8.6, 9.5, 19.4, 9.5, 5.9]$ p.u., which is the maximum amount allowed in the first attack scenario. In both scenarios, the attack lasts during $0$-$20$s and terminates ($\mathbf{P}^a = 0$) after $20$s.
Fig.~\ref{fig:attack_ne_none.1} shows the strategic attack in the first scenario. Although the attacker has caused a large frequency deviation for all generators at the beginning, the defender manages to stabilize the system quickly and all generators' frequencies gradually converge to the nominal $\omega_{\mathrm{n}}$. Besides, the system stabilizes more quickly compared with Fig.~\ref{fig:nodefense.1}, showing the advantage of the proposed strategic resilient control.
The results of the second attack scenario are depicted in Fig.~\ref{fig:attack_ne_none.2}. Since the attack action is not strategic, we observe a smaller overshoot in the frequency deviation at around $5$s. A frequency ripple happens around $20$s because the attacker suddenly shuts down all the manipulated load. However, the defender's action can quickly adapt to the load change and stabilize the system. It shows that the min-max controller enables the defender to combat both strategic and non-strategic attacks, significantly improving grid resilience.

\begin{figure}[!t]
    \centering
    \subfigure[Strategic attack.]{\includegraphics[width=0.48\columnwidth]{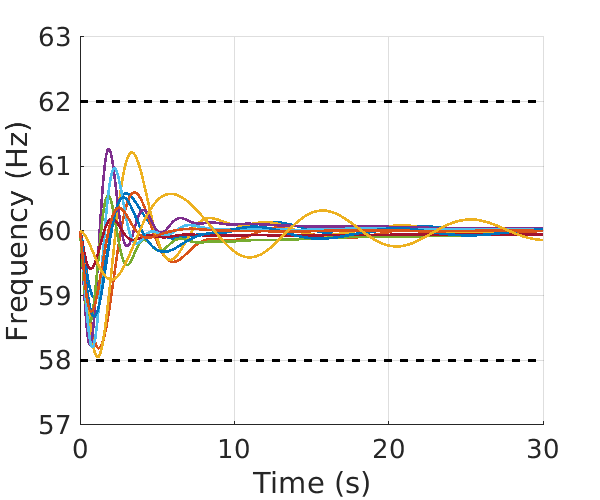}\label{fig:attack_ne_none.1}}
	\subfigure[Constant load attack.]{\includegraphics[width=0.48\columnwidth]{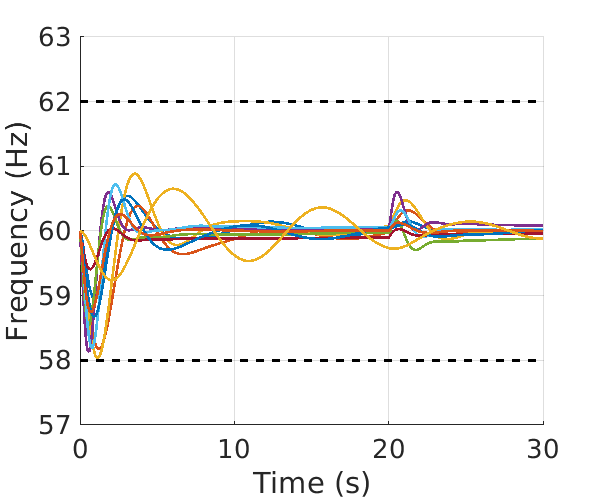}\label{fig:attack_ne_none.2}}  
\caption[]{(a) implements the strategic attack, where the attacker manipulates the load strategically to destabilize the system. (b) shows a constant load attack. The attacker turns on the loads at $t=0$s and shuts them down at $t=20$s. The defender can stabilize the system in both attack scenarios with the designed resilient control.}
\label{fig:attack_ne_none}
\end{figure}

\subsection{Dynamic Defense for Agile Cyber-Physical Resilience} \label{sec:case.dynamic}
We implement a dynamic botnet attack to demonstrate holistic and agile cyber-physical grid resilience under the proposed approach. The dynamic botnet attack is specified as follows.
\begin{itemize}
    \item Attack stage 1: The attacker uses the same setting as in Sec.~\ref{sec:case.cyber} to initiate the botnet attack and perform the strategic attack to the grid for $10$s.
    \item Attack stage 2: The attacker's cyberattack intensity is dropped to $\zeta(u_a) = 1.5\log(u_a) + 0.1$; the attacker performs load switching attack in the grid for $10$s, i.e., turning on $0.9\cdot \mathbf{P}^{LV}$ in the first $5$s and turning them off in the second $5$s.
    \item Attack stage 3: The defender's cyber defense cost is increased to $C_d(u_d) = 0.3 u_d^2$; the attacker performs the strategic attack to the grid for $10$s.
\end{itemize}

The cyber risk $I(t)$ in Fig.~\ref{fig:dyn_attack.1} quickly reaches the steady state in three attacks because the cyber time-scale is much faster than the physical counterpart. $I(t)$ drops at attack stage 2 because the attacker has a weaker attack intensity. It rises in attack stage 3 because of the increase in the defender's cyber defense cost. The steady state cyber risks $\bar{I}$ in all stages are $[0.56, 0.36, 0.46]$. Hence, the systemic risk $\bar{R}$ in all stages are $[166.4, 107.8, 139.8]$ p.u., providing different constraints in physical system regulation.
Fig.~\ref{fig:dyn_attack.2} shows the generator frequency deviation in each attack stage.
We also plot the defender's regulation power $\mathbf{P}^d$ and the attacker's manipulated load $\mathbf{P}^a$ for selected generator and load buses in Fig.~\ref{fig:dyn_action} for better visualization.

\begin{figure}[!t]
    \centering
    \subfigure[Cyber risk evolution.]{\includegraphics[width=0.48\columnwidth]{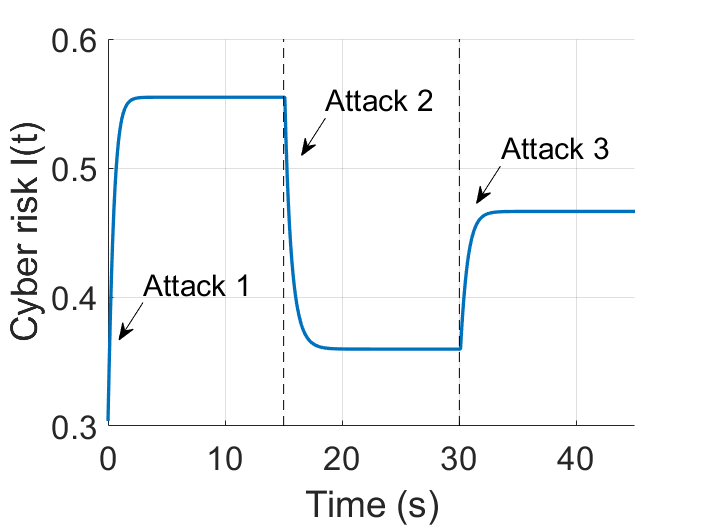}\label{fig:dyn_attack.1}}
	\subfigure[Physical system evolution.]{\includegraphics[width=0.48\columnwidth]{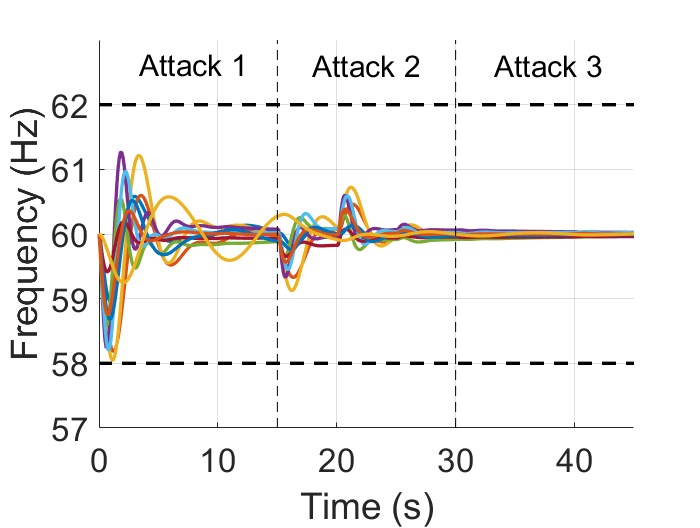}\label{fig:dyn_attack.2}}  
\caption[]{Cyber risk and generator frequency evolution under dynamic attacks. In each attack stage, the total amount of vulnerable load changes because of the variation in cyber defense strategy, which affects the dynamic power regulation schemes at the physical layer. The defender manages to stabilize the grid under all attacks.}
\label{fig:dyn_attack}
\end{figure}

\begin{figure}[!t]
    \centering
    \subfigure[Defender's strategy $\mathbf{P}^d$]{\includegraphics[width=0.48\columnwidth]{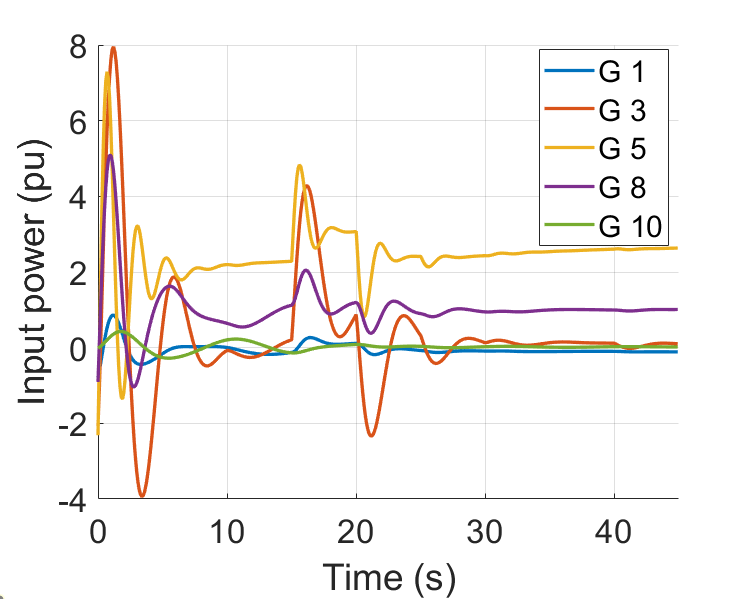}\label{fig:dyn_action.1}}
	\subfigure[Attacker's strategy $\mathbf{P}^a$]{\includegraphics[width=0.48\columnwidth]{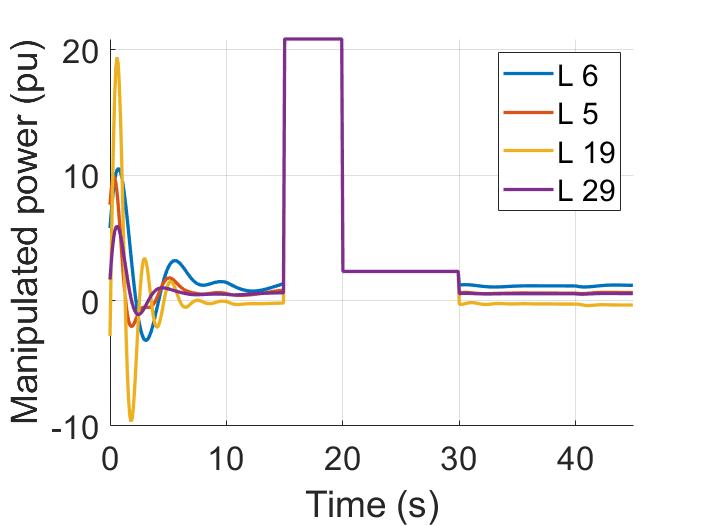}\label{fig:dyn_action.2}}  
\caption[]{(a) Selected defender's resilient control on power regulation and (b) the attacker's load manipulation under the investigated dynamic botnet attacks.} 
\label{fig:dyn_action}
\end{figure}

Fig.~\ref{fig:dyn_attack.2} indicates that the defender can effectively stabilize the system for all attacks. The frequency fluctuates only when an attack starts, and the defender can quickly mitigate the attack, showing the agile resilience of the grid. It is worth noting that once the grid is stabilized, strategic attacks can have little effect in disrupting the grid, which can be observed in the last half of attack stage 1 ($5$-$10$s) and the attack stage 3 ($30$-$40$s) in Fig.~\ref{fig:dyn_action.2}. In attack stage 1, after the system is stabilized (around $5$s), the attacker begins to reduce the manipulation of vulnerable loads despite having larger access to those loads. 
There is a ripple in the grid frequency because of the sudden load change. The operator also changes the regulation strategy to protect the grid, and the frequency quickly converges to $\omega_n$, as shown in attack stage 2. We note that the frequency changes slightly in attack stage 2 compared to attack stage 1. This is because the operator has already used strategic regulation strategies from attack stage 1, which further mitigates the attack consequence.
In attack stage 3, the attacker only maliciously controls a small portion of loads up to $5.33$ p.u. despite the total vulnerable loads being $139.8$ p.u. This is because the defender can easily regulate the grid and maintain its stability. The attacker needs to cause more disruption that exceeds the defender's regulation capacity to affect the grid's normal operation, which is not cost-effective for strategic attackers. In this situation, as shown in attack stage 2, only irrational attacks can disrupt the system. But the defender can quickly stabilize such disturbance using the proposed min-max controller. In summary, our developed control scheme enhances the agile cyber-physical resiliency of the power grids under strategic, non-strategic, and consecutive attacks.

\section{Conclusion}\label{sec:conclusion}
With the universal adoption of IoT-controlled high-power energy devices in households, the cybersecurity of modern power grids is a critical concern. We have investigated the IoT botnet attack in which the adversary controls loads of the grid by manipulating IoT energy devices dynamically. The developed epidemic model has provided a tractable solution to quantify the systemic cyber risks of power grids. The cross-layer game-theoretic cyber defense mechanism and physical resilient control have been shown effective in maintaining the grid's normal operation under the considered strategic botnet attack and hence have improved the grid's integrative resiliency at both the cyber and physical layers. 
For future work, we would conduct detailed device-level simulations to improve attack process resolutions and explore more practical large-scale modeling approaches to characterize and combat botnet attacks in power grids. We would also investigate the scenario when the system operator has unknown information on the IoT botnet attacker's model and objective and develop learning-based cyber-physical resilience countermeasures.

\bibliographystyle{IEEEtran}
\bibliography{IEEEabrv,references}

\end{document}